\documentclass[sigconf,natbib=true]{acmart}

\usepackage[ruled,vlined,linesnumbered]{algorithm2e}
\usepackage{graphicx}
\usepackage{capt-of}    
\usepackage{tabularx}   %
\usepackage{multicol}
\usepackage{multirow}

\usepackage[export]{adjustbox}
\usepackage{subfig}

\AtBeginDocument{%
  \providecommand\BibTeX{{%
    \normalfont B\kern-0.5em{\scshape i\kern-0.25em b}\kern-0.8em\TeX}}}

\copyrightyear{2022}
\acmYear{2022}
\setcopyright{acmcopyright}\acmConference[SIGIR '22]{Proceedings of the 45th
International ACM SIGIR Conference on Research and Development in Information
Retrieval}{July 11--15, 2022}{Madrid, Spain}
\acmBooktitle{Proceedings of the 45th International ACM SIGIR Conference on
Research and Development in Information Retrieval (SIGIR '22), July 11--15, 2022,
Madrid, Spain}
\acmPrice{15.00}
\acmDOI{10.1145/3477495.3531959}
\acmISBN{978-1-4503-8732-3/22/07}

\newcommand{\datasetA}{MovieLens\xspace}
\newcommand{\datasetB}{Epinion\xspace}
\newcommand{\datasetC}{Gowalla\xspace}
\newcommand{\datasetD}{LastFM\xspace}

\newcommand{\usergroupA}{active\xspace}
\newcommand{\usergroupB}{inactive\xspace}

\newcommand{\itemgroupA}{short-head\xspace}
\newcommand{\itemgroupB}{long-tail\xspace}

\newcommand{\dquotes}[1]{``#1''}
\newcommand{\squotes}[1]{`#1'}

\newcommand{\eg}{e.g., }
\newcommand{\ie}{i.e., }
\newcommand{\wrt}{w.r.t.~}

\newcommand{\partitle}[1]{\vspace{2mm}\noindent\textbf{#1}}

\DeclareMathOperator{\E}{\mathbb{E}}

\begin{document}

\title{CPFair: Personalized Consumer and Producer Fairness Re-ranking for Recommender Systems}


\author{Mohammadmehdi Naghiaei}
\affiliation{%
  \institution{University of Southern California}
  \city{California}
  \country{USA}}
\email{naghiaei@usc.edu}

\author{Hossein A.~Rahmani}
\affiliation{%
  \institution{University College London}
  \city{London}
  \country{United Kingdom}}
\email{h.rahmani@ucl.ac.uk}

\author{Yashar Deldjoo}
\affiliation{%
  \institution{Polytechnic University of Bari}
  \city{Bari}
  \country{Italy}}
\email{yashar.deldjoo@poliba.it}

\renewcommand{\shortauthors}{M.~Naghiaei, H.~A.~Rahmani, Y.~Deldjoo}

\begin{abstract}
Recently, there has been a rising awareness that when machine learning (ML) algorithms are used to automate choices, they may treat/affect individuals unfairly, with legal, ethical, or economic consequences. Recommender systems are prominent examples of such ML systems that assist users in making high-stakes judgments.

A common trend in the previous literature research on fairness in recommender systems is that the majority of works treat user and item fairness concerns separately, ignoring the fact that recommender systems operate in a two-sided marketplace. In this work, we present an optimization-based re-ranking approach that seamlessly integrates fairness constraints from both the consumer and producer-side in a joint objective framework. We demonstrate through large-scale experiments on 8 datasets that our proposed method is capable of improving both consumer and producer fairness without reducing overall recommendation quality, demonstrating the role algorithms may play in minimizing data biases.
\end{abstract}




\maketitle

\section{Introduction}
\label{sec:introduction}

\begin{figure}[t]
    \centering
    \includegraphics[scale=0.4]{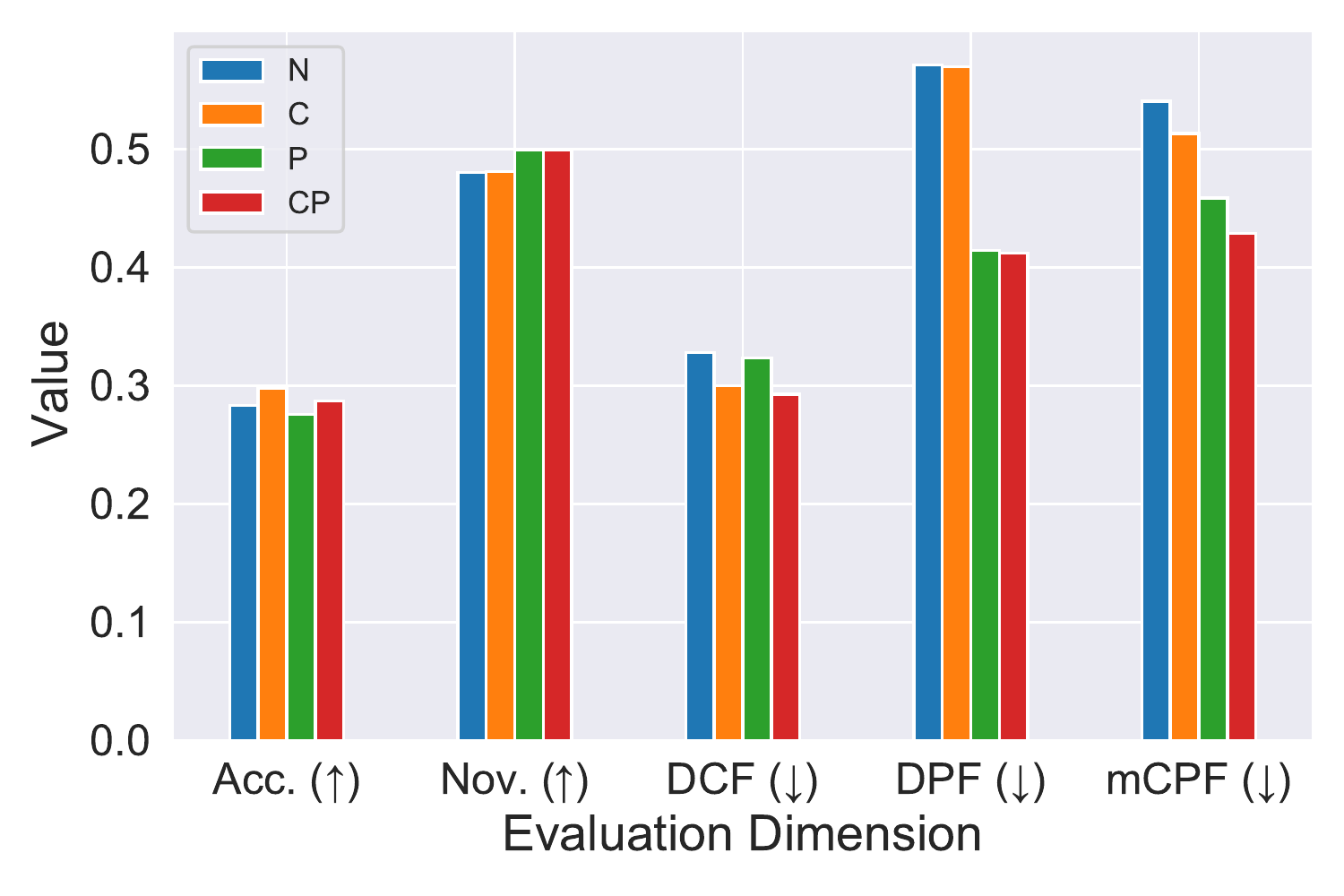}\caption{Fairness-unaware (N), user-oriented (C), item-oriented (P), and two-sided fairness-performance (CP) of recommendation algorithms on accuracy, novelty, and fairness performance. Consumer fairness evaluation (DCF) and producer fairness evaluation (DPF) are the two metrics that make up the mCPF. Note that CP represents the core of our contribution. The results on each bar show the average of 32 experiments across datasets and baseline CF models.}
    \label{fig:introresults}
\end{figure}

Recommender systems are ubiquitous and support high-stakes decisions in a variety of application contexts such as online marketing, music discovery, and high-stake tasks (\eg job search, candidate screening, and loan applications). By delivering personalized suggestions, these systems have a major impact on the content we consume online, our beliefs and decisions, and the success of businesses. Recently, there has been a growing awareness about the fairness of machine learning (ML) models in automated decision-making tasks, such as classification \cite{mehrabi2021survey,binns2018fairness}, and ranking/filtering tasks \cite{chen2020bias,caton2020fairness}, with recommender systems serving as a notable example of the latter. However, unlike classification, fairness in recommender systems (RS) is a multi-faceted subject, depending on stakeholder, type of benefit, context, morality, and time among others~\cite{ekstrand2021fairness,deldjoo2021flexible,burke2017multisided}.

A recurring theme in the fairness recommendation literature is that computing research is often focused on a particular aspect of fairness, framing the problem as a building algorithm that fulfills a specific criterion, such as optimizing for producer fairness~\cite{dong2020user,yalcin2021investigating} or consumer fairness~\cite{dwork2012fairness,islam2021debiasing}. However, this strategy may appear oversimplified, requiring more nuanced and multidisciplinary approaches to research recommendation fairness.

One overarching aspect that may help to unify studies on fairness in multi-stakeholder settings is distinguishing the \textit{(i)} benefit type (exposure vs.~relevance), and (\textit{ii)} the main stakeholders involved in the recommendation setting (consumers vs. producers). Exposure refers to how evenly items or groups of items are exposed to users or groups of users. Relevance (effectiveness) determines to what extent the items exposition is effective, \ie matches with the user's taste. The second level of the taxonomy is stakeholder. Almost every online platform we interact with (\eg Spotify, Amazon) serves as a marketplace connecting consumers and item producers/service providers, making them the primary end-beneficiaries and stakeholders in RS. From the consumers' perspective, fairness is about distributing effectiveness evenly among users, whereas producers and item providers seeking increased visibility are primarily concerned with exposure fairness. 

\begin{figure}
    \centering
    \includegraphics[scale=0.5]{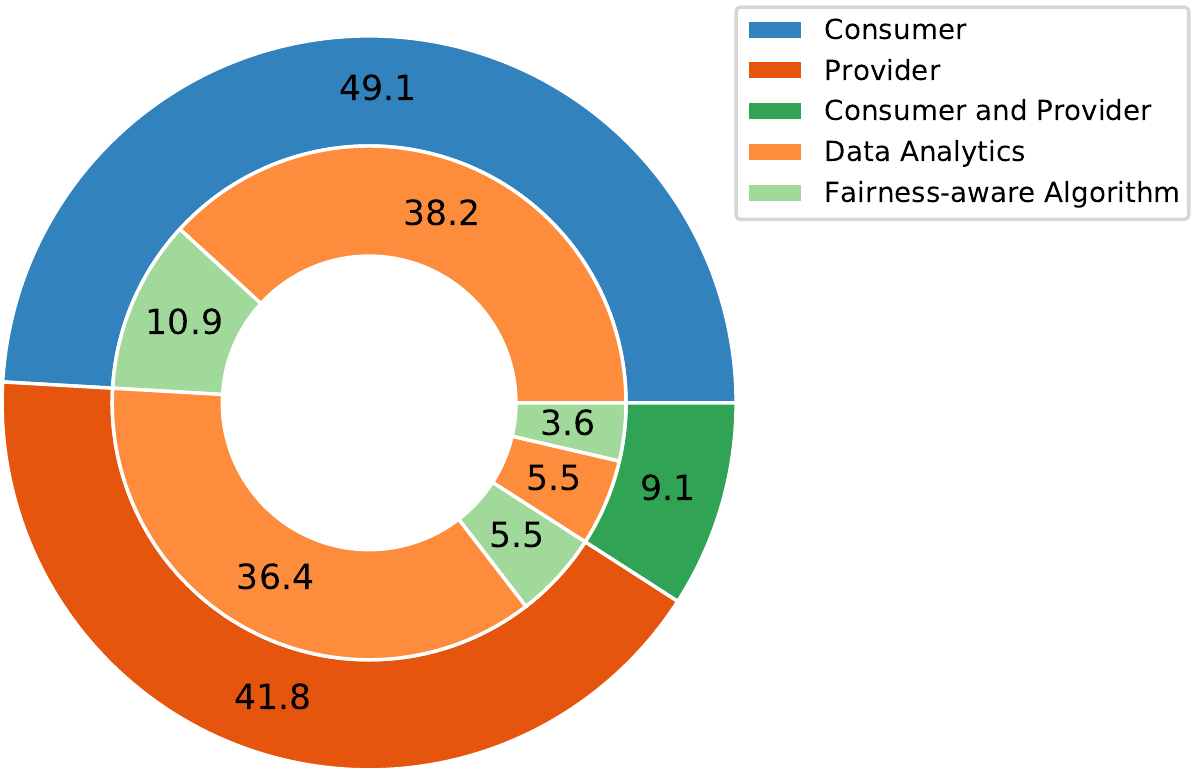}
    \caption{The percentage of the research studied different aspects of fairness in recommender systems.}
    \label{fig:paperstatistics}
\end{figure}

Figure \ref{fig:paperstatistics} illustrates the distribution of research on various aspects of fairness in recommender systems.\footnote{The figures are based on 120 publications retrieved from DBLP using the keywords ``fair/biased recommendation'', ``re-ranking'', and ``collaborative filtering''.} We may observe a division/split in the research on fairness-aware recommendation algorithms, with around 49.1\% of papers concentrating on consumer fairness and somewhat fewer on producer fairness (41.8\%). Few studies (less than 10\%) address consumer and producer fairness concerns simultaneously. However, the underlying user-item interaction coexists on (and can impact) both sides of beneficiary stakeholders. For example, there may be disparities in the consumption of item groups (defined by protected attributes) between active and inactive users in specific domains. Active users of POI may want to visit a variety of popular and less popular locations, whilst less active users may want to visit mainly the most popular locations. Prior research on unilateral fairness optimization raises various questions, including whether optimizing for one party's advantages, such as consumers, may ensure/protect the interests of other parties (\ie producers) and vice versa? Furthermore, how does optimizing for one stakeholder's benefits (\eg consumers' relevance) affect overall system accuracy? Are there ways to explain the key interplay and trade-offs between fairness restrictions (user relevance, item exposure) and overall system accuracy in various recommendation scenarios? Are there ways to explain the key interplay and trade-offs between fairness restrictions (user relevance, item exposure) and overall system accuracy in various recommendation scenarios?

To fill this gap and address the above questions, in this work, we focus our attention on fairness in a \textit{marketplace} context, or more accurately, \textit{multi-stakeholder setting}, and aim to unify the two aspects of fairness, consumer and provider fairness. To this end, we propose a model-agnostic re-ranking approach capable of jointly optimizing the \textit{marketplace} objectives, \ie giving fair exposure to producers' items while minimizing disparate treatment against consumer groups without considerable sacrifice on recommender system total accuracy.

To better illustrate the capabilities of our proposed system, in Figure~\ref{fig:introresults}, we compare the performance of various recommendation algorithms, against unilateral (or single-sided) fairness objectives according to the perspectives of different beneficiary stakeholders.

In particular, (N) serves as the baseline fairness unaware recommendation algorithm, whereas (C), (P), and (CP) are three variations of our proposed system, taking into account consumer- and provider-side fairness objectives separately, and both simultaneously. (C) closely resembles the approach proposed by~\citet{li2021user} with the difference that in our formulation, the users' parity objective have been integrated into the system in the form of an unconstrained optimization problem, and a greedy algorithm is used to solve it optimally in polynomial time.  

Figure~\ref{fig:introresults} represents the average results of the 128 experimental cases examined in this work, \ie $8 \ \text{(datasets)} \ \times \ 4 \ \text{(CF baselines)} \ \times \ 4 \ (\text{fairness constraints type})$. One may observe the drawbacks of unilateral fairness optimization, in which DCF and DPF are used to measure unfairness (biases) on the consumer and supplier sides, respectively (cf. Section~\ref{subsec:market_obj}). For example, user-oriented fairness optimization (C-fairness) can enhance both consumer fairness (\ie reducing DCF) and total system accuracy. However, this creates a huge bias against suppliers. Producer-fairness optimization, on the other hand, can improve novelty and producer fairness (\ie reducing DPF), indicating that users' recommendation lists contain a greater number of long-tail items. However, this comes at the cost of a considerable reduction in overall accuracy (from $0.2832$ to $0.2756$) and worsening of DCF (from $0.3280$ to $0.3234$). 
The proposed consumer-provider fairness optimization (CP-fairness) in the present work combines the benefits of previous system variants, enhancing both producer and consumer fairness, resulting in increased accuracy (from $0.2832$ to $0.2868$) and even increased novelty (due to recommendation of more long-tail items). To summarize, the key contributions of this work are as follow:
\begin{itemize}
    \item \textbf{Motivating multi-stakeholder fairness in RS.} We motivate the importance of multi-sided fairness optimization by demonstrating how inherent biases in underlying data may impact both consumer and producer fairness constraints, as well as the overall system's objectives, \ie accuracy and novelty (cf.~Section~\ref{sec:motivating_cp}); these biases and imbalances if left unchecked, could lead to stereotypes, polarization of political opinions, and the loss of emerging business. 
    \item \textbf{CP-Fairness modeling.} Our research builds on and expands earlier work \cite{li2021user}, in several dimensions: \textit{(i)} we consider \textit{multi-stakeholder} objectives including both consumers' and producers' fairness aspects; \textit{(ii)} we formalize the re-ranking problem as an integer programming method without enforcing the fairness objectives as a constraint that is susceptible to reduction in overall accuracy in a \textit{marketplace}; \textit{(iii)} we propose an efficient greedy algorithm capable of achieving an optimal trade off within \textit{multi-stakeholder} objectives in a polynomial-time (cf.~Section~\ref{sec:proposed_method}). In summary, this makes the proposed contribution more versatile and applicable to various recommendation models and settings irrespective of their learning criteria and their recommendation list size.
    \item \textbf{Experiments.} We conduct extensive experimental evaluations on 8 real-world datasets from diverse domains (\eg Music, Movies, E-commerce), utilizing implicit and explicit feedback preferences (cf.~Section~\ref{sec:experimental_methodology}); this, combined with four recommendation algorithm baseline amount the number of experiments carried out to 128 recommendation simulation.
\end{itemize}

\section{Background and Related work}
\label{sec:related_work}
A promising approach to classifying the literature in recommendation fairness is according to the beneficiary stakeholder~\cite{abdollahpouri2020multistakeholder,ekstrand2021fairness}. Based on our survey screening results on papers published in fairness during the past two years, only 9.1\% of publications in the field deal with multi-sided fairness, with even a less percentage $3.6\%$ proposing an actual two-sided fairness-aware algorithm, as pursued in our current study.

\partitle{C-fairness Methods.}
\citet{dwork2012fairness} introduces a framework for individual fairness by including all users in the protected group. \citet{abdollahpouri2019unfairness} and \citet{naghiaei2022unfairness} investigate a user-centered evaluation of popularity bias that accounts for different levels of interest among users toward popular items in movies and book domain. \citet{abdollahpouri2021user} propose a regularization-based framework to mitigate this bias from the user perspective. \citet{li2021user} address the C-fairness problem in e-commerce recommendation from a group fairness perspective, \ie a requirement that protected groups should be treated similarly to the advantaged group or total population \cite{pedreschi2009measuring}. Several recent studies have indicated that an exclusively consumer-centric design approach, in which customers' satisfaction is prioritized over producers' interests, may result in a reduction of the system's overall utility~\cite{patro2020fairrec,wang2021user}. As a result, we include the producer fairness perspectives in our proposed framework as the second objective in this work.

\partitle{P-fairness Methods.}
Several works have studied recommendation fairness from the producer's perspective \cite{yalcin2021investigating,wundervald2021cluster}. \citet{gomez2022provider} assess recommender system algorithms disparate exposure based on producers' continent of production in movie and book recommendation domain and propose an equity-based approach to regulate the exposure of items produced in a continent. \citet{dong2020user} set a constraint to limit the maximum times an item can be recommended among all users proportional to its popularity to enhance the item exposure fairness. In contrast, \citet{ge2021towards} investigate fairness in a dynamic setting where item popularity changes over time in the recommendation process and models the recommendation problem as a constrained Markov Decision Process with a dynamic policy toward fairness.

A common observation in C-fairness and P-Fairness research is the type of attributes used to segment the groups on consumer and provider side. These attributes could be internal, e.g., determined via user-item interactions~\cite{li2021user,abdollahpouri2019managing} or provided externally, e.g., protected attributes such as gender, age, or geographical location~\cite{gomez2022provider,boratto2021interplay}. In this work, we segmented the groups by internal attributes and based on the number of interactions for both consumers and providers.

\partitle{CP-fairness Methods.}
Considering both consumers' and providers' perspectives, \citet{chakraborty2017fair} present mechanisms for CP-fairness in matching platforms such as Airbnb and Uber. \citet{rahmani2022unfairness} studied the interplays and tradeoffs between consumer and producer fairness in Point-of-Interest recommendations. \citet{patro2020fairrec} map fair recommendation problem to the constrained version of the fair allocation problem with indivisible goods and propose an algorithm to recommend top-K items by accounting for producer fairness aspects. \citet{wu2021tfrom} focus on two-sided fairness from an individual-based perspective, where fairness is defined as the same exposure to all producers and the same nDCG to all consumers involved.

\citet{do2021two} define the notion of fairness in increasing the utility of the worse-off individuals following the concept of distributive justice and proposed an algorithm based on maximizing concave welfare functions using the Frank-Wolfe algorithm. \citet{lin2021mitigating} investigate sentiment bias, the bias where recommendation models provide more relevant recommendations on user/item groups with more positive feedback, and its effect on both consumers and producers.

The inner layer of the circle in Figure~\ref{fig:paperstatistics} presents approaches for recommendation fairness that could be classified according to 1) developing metrics to quantify fairness, 2) developing frameworks for producing fair models (according to the desired notion of fairness), or 3) modifying data to combat historical bias. Contrary to the works surveyed, we develop a fairness-aware algorithm that simultaneously addresses user groups and item group fairness objectives jointly via an efficient optimization algorithm capable of increasing the system's overall accuracy.
\section{Motivating CP-Fairness Concerns}
\label{sec:motivating_cp}

\begin{table}
    \centering
    \caption{Percentage of users and items located at different number of interactions thresholds (as $n$ and $r$ represents, respectively) in the training set of the datasets.}
    \begin{adjustbox}{max width=240pt}
    \begin{tabular}{l|llll|llll}
    \toprule
        Dataset & $n\geq10$ & $n\geq20$ & $n\geq50$ & $n\geq100$ & $r\geq10$ & $r\geq20$ & $r\geq50$ & $r\geq100$  \\ \midrule
        MovieLens & 100\% & 82.08\% & 46.98\% & 26.09\% & 77.39\% & 60.49\% & 34.54\% & 16.46\% \\
        Epinion & 99.85\% & 56.60\% & 8.11\% & 1.57\% & 99.81\% & 63.83\% & 17.33\% & 4.61\% \\
        Gowalla & 100\% & 86.11\% & 25.40\% & 4.87\% & 100\% & 84.44\% & 20.10\% & 3.95\% \\
        LastFM & 96.33\% & 72,73\% & 0.00\% & 0.00\% & 73.99\% & 37.36\% & 15.00\% & 4.48\%\\
    \bottomrule
    \end{tabular}
    \label{tbl:dataset_charactristics}
    \end{adjustbox}
\end{table}

In this section, we intend to motivate the need of having two-sided fairness (\ie user and item) by undertaking both data and algorithm analysis. We analyze the distribution and properties of two well-known real-world recommendation datasets, \ie \datasetB and \datasetC, with their details summarized in Table~\ref{tbl:datasets}. Due to space consideration, we only show the analysis on these datasets; however, we observed similar patterns on the other datasets as can be seen in Figure \ref{fig:CPevalBoxPlot} baseline (N) models.

In Table~\ref{tbl:dataset_charactristics}, we show the distribution of users and items with different numbers of interactions in the datasets. As the values show, most users are concentrated in areas with less interaction with the items. We can also note that majority of items have received fewer interactions from the users (\ie long-tail items), and a small percentage of items (\ie popular items) are recommended frequently. This motivates the need to expand the long-tail coverage of recommendation lists while encouraging users with fewer interactions to engage with the system.

\begin{figure*}
  \centering
  \subfloat[Epinion]
    {\includegraphics[scale=0.24]{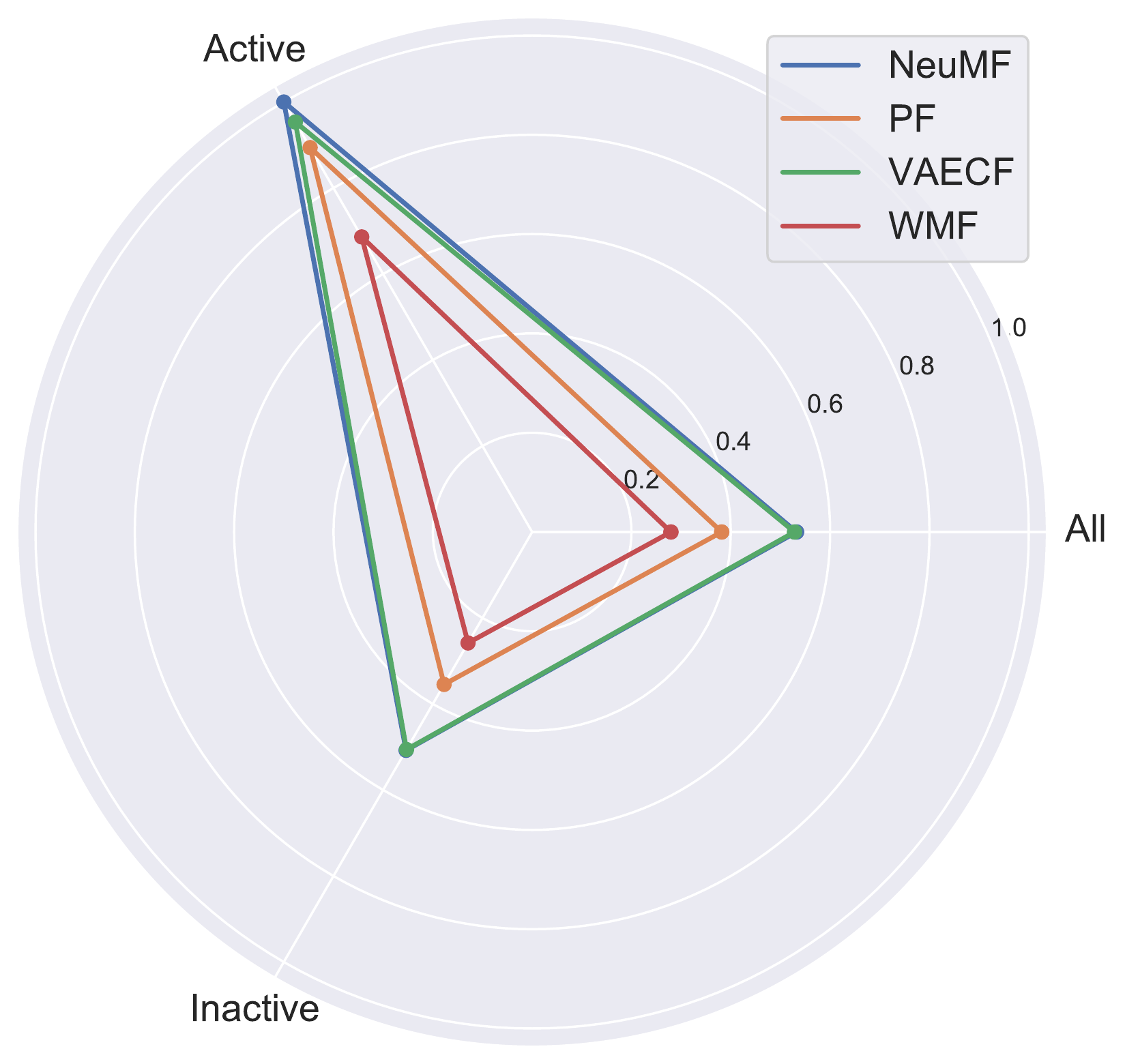}
    \label{fig:radar_user_epinion}}
  \hfill
  \subfloat[Gowalla]
    {\includegraphics[scale=0.24]{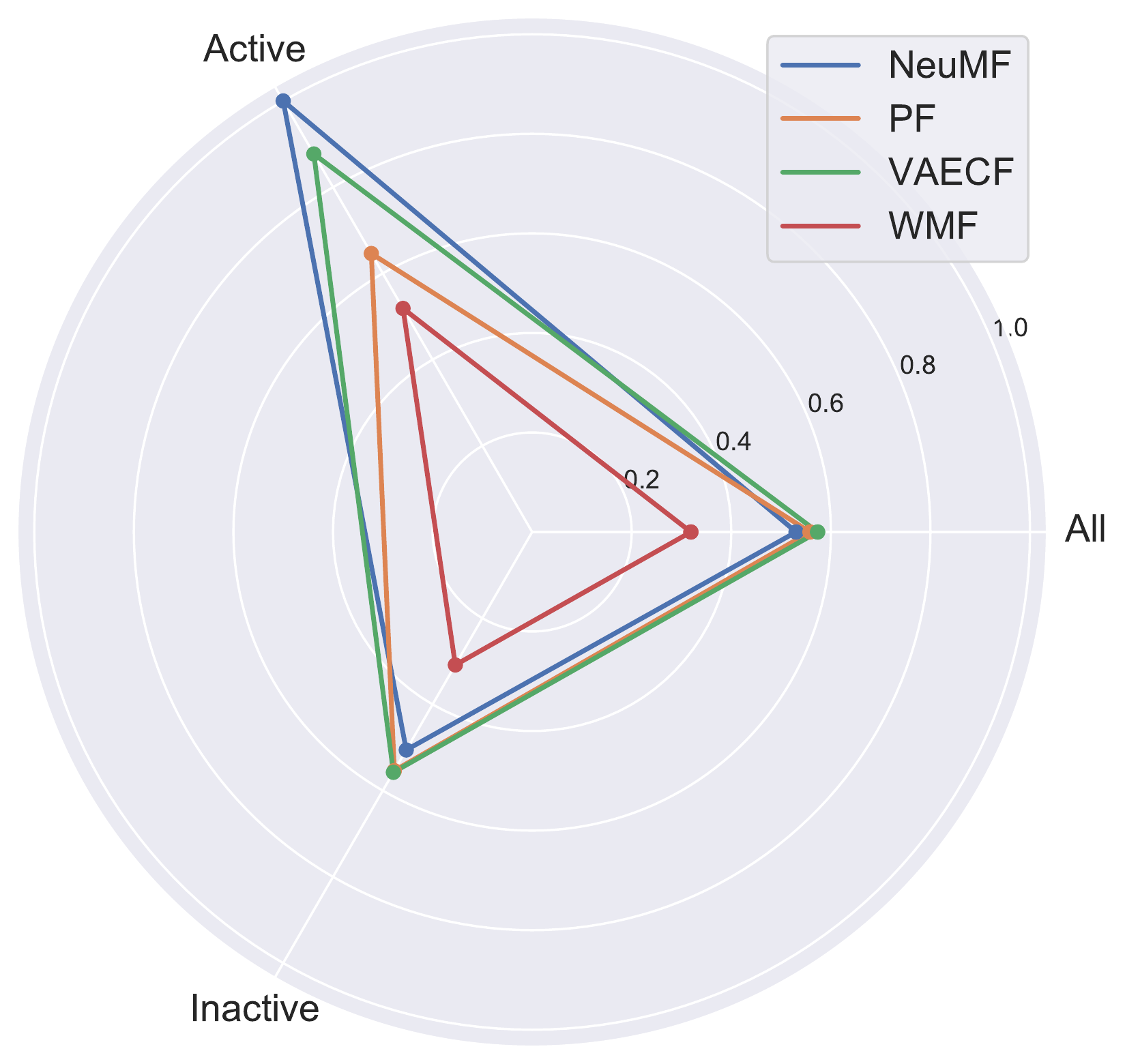}
    \label{fig:radar_user_gowalla}}
  \hfill
  \subfloat[Epinion]
    {\includegraphics[scale=0.24]{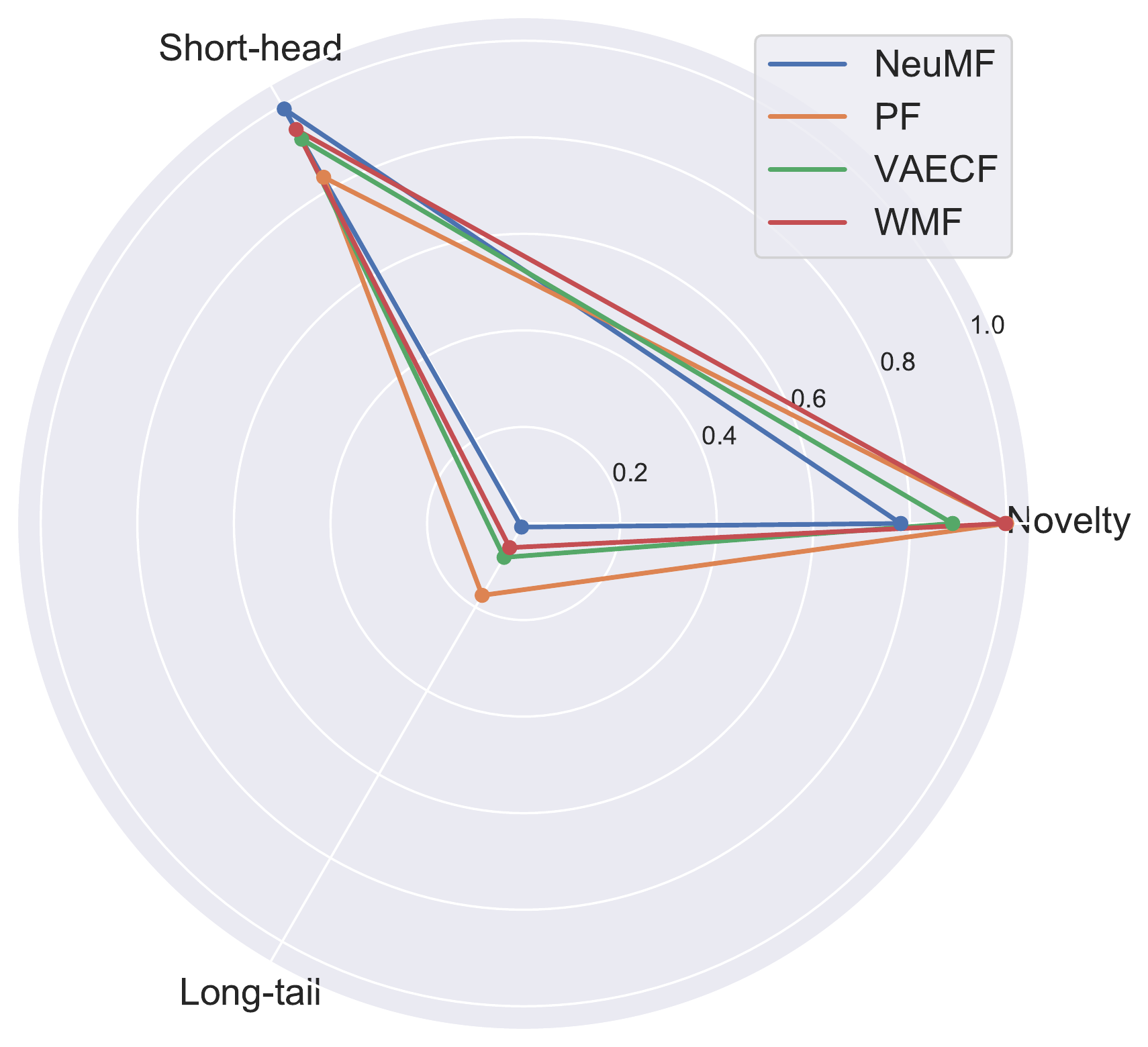}
    \label{fig:radar_item_epinion}}
  \hfill
  \subfloat[Gowalla]
    {\includegraphics[scale=0.24]{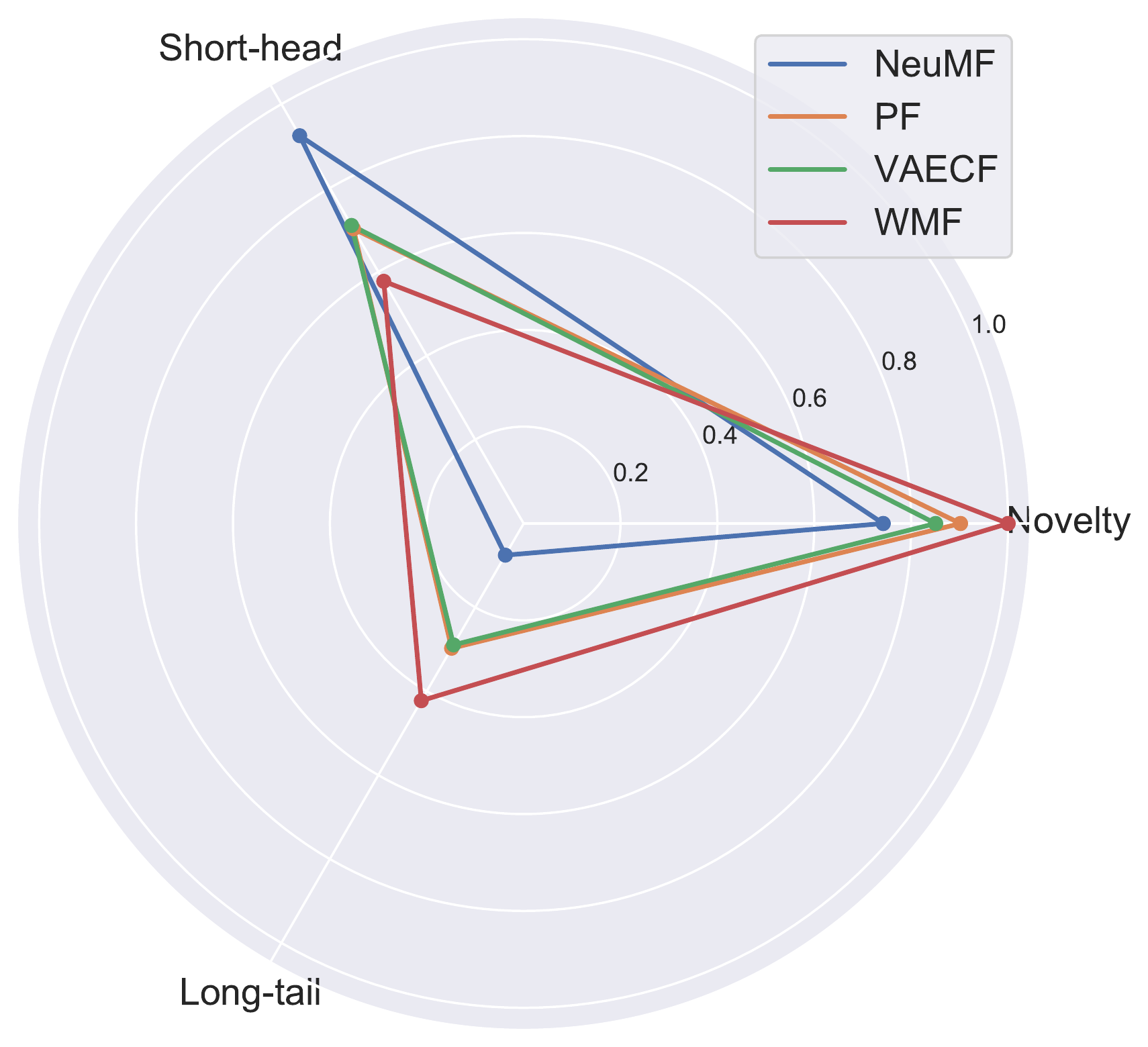}
    \label{fig:radar_item_gowalla}}
  \caption{Figures (a) and (b) show the difference between \usergroupA and \usergroupB user groups on normalized nDCG@10 and Figures (c) and (d) indicate the difference between \itemgroupA and \itemgroupB item groups on item exposure and novelty.}
\label{fig:unfairness_rec}
\end{figure*}

In the next step, we analyze the recommendation quality of the fairness-unaware baseline algorithms on the two datasets for both user and item groups. In Figure \ref{fig:unfairness_rec}, we show the min-max normalized values of the recommendation algorithms quality on the \datasetB and \datasetC datasets. In Figures \ref{fig:radar_user_epinion} and \ref{fig:radar_user_gowalla} we plot the accuracy (\ie nDCG@10) of all users, \usergroupA users, and \usergroupB users group. Later, in Section \ref{sec:results}, we present the details of the experiments results in Table \ref{tbl:results_epinion}.
We see that the \usergroupA group, representing a small proportion of the total users, receives a much higher recommendation quality than the \usergroupB group. Based on this finding, commercial recommendation engines have a tendency to disregard the majority of users. For example, from Figure \ref{fig:radar_user_epinion}, we see that NeuMF and VAECF almost achieve the best performance on \usergroupA user group ($1.0$) while for \usergroupB user group this number is almost $0.6$. Thus, this degrades the overall performance, which is only slightly above the inactive groups at about $0.6$ on both recommendation algorithms. Moreover, in Figures \ref{fig:radar_item_epinion} and \ref{fig:radar_item_gowalla}, we plot the number of \itemgroupA and \itemgroupB items that are recommended to the users as well as the novelty of the recommendation. As can be seen, although the number of items in \itemgroupA group is very small in comparison to \itemgroupB group, a large proportion of recommended items are short-head (\ie popular) items. Therefore, we can note that the unfairness of recommendation on the provider-side (measured in terms of exposure) is even higher than the preceding user-relevance fairness. Our results show that the fairness of recommendation is essential from both perspectives, \ie consumers and providers. Hence, it is important to propose techniques to better serve disadvantaged users and item groups.
\section{Proposed CP-fairness algorithm}
\label{sec:proposed_method}
In this section, we first define fairness in a multi-sided platform based on deviation from the parity from consumers' and producers' perspectives. Then we propose a mixed-integer linear programming re-ranking algorithm to generate a fair list of recommendations for each user given a traditional unfair recommendation list.

\partitle{Problem Formulation.}
Let $\mathcal{U}$ and $\mathcal{I}$ be the set of consumers and items with size $n$ and $m$, respectively. In a traditional recommender system application, a user, $u \in \mathcal{U}$, is provided a list of top-N items, $L_N(u)$ according to some relevance scores, $s \in S^{n\times N}$. The goal is to compute a new fair ordered list, $L^{F}_K(u)$ from the original unfair baseline recommended items $L_N(u)$ where $L^{F}_K(u) \subset L_N(u) \ \ where \ \ K \leq N$ such that the new list would be more appealing to the system as a whole considering the relevance score to the user $u$ and fairness objectives jointly.

Formally, item $I_x$ is presented ahead of item $I_y$ in $L_N(u)$ if $s_{ux} > s_{uy}$. Recommending the items to users based solely on the relevance score is designed to help improve the overall system recommendation accuracy regardless of the fairness measures in a multi-sided platform. We define a binary decision matrix $A=[A_{ui}]_{n \times N}$ with value 1 when the item ranked $i$ in $L_N(u)$ is recommended to user $u$ in the new fairness-aware recommendation list, i.e., $L^F_K(u)$ and 0 otherwise. Hence, the fair top-K recommendation list to user $u$, can be represented as $[A_{u1},A_{u2},...,A_{uN}]$ when $\sum_{i=1}^{N}{A_{ui}} = K $, $K\leq N$.
We discuss below an  optimization framework to calculate optimal values for matrix $A$ and re-rank original unfair recommendation list.

\subsection{Marketplace Fairness Objectives}
\label{subsec:market_obj}
The fairness representation target, \ie the ideal distribution of exposure or relevance among groups, is the first step to designing a fair recommendation model that should be decided based on the application and by the system designers. \citet{kirnap2021estimation} classified the fairness notions commonly used in IR literature into three groups based on their ideal target representation, (1) parity between groups; (2) proportionality in exposure in groups according to the group/catalog size; and (3) proportionality to the relevance size to the recommended groups. This work focuses on the notion of fairness defined based on a \dquotes{parity-oriented} target representation, leaving other concepts of fairness for future investigations.

\subsubsection{Fairness constraint on producer side} 
On the producer side, fairness may be quantified in terms of two distinct types of benefit functions: exposure and relevance, both of which can be binary or graded \cite{gomez2022provider}. Given that suppliers' primary objective is their products' visibility, this study focuses its attention exclusively on the (binary) exposure dimension, leaving the other producer fairness constraints for future work. To measure the exposure parity, we divide items into two categories of short-head $\mathcal{I}_1$ and long-tail $\mathcal{I}_2$ defined based on their popularity score where $\mathcal{I}_1 \cap \mathcal{I}_2 = \emptyset $. 
We define the notation of $\mathcal{MP}$ as a metric that quantifies item groups' exposition in the recommendation list decision, $A$, such as item group exposure count as used in our experiments. Deviation from Producer Fairness (DPF) is the metric that defines the deviation from parity exposition defined by:

\begin{equation}
    \mathbf{DPF}(A, \mathcal{I}_1, \mathcal{I}_2) = \E\left[\mathcal{MP}(i,A)| i \in \mathcal{I}_1 \right] - \E\left[\mathcal{MP}(i,A)|i \in \mathcal{I}_2\right] 
    \label{eq:dpf}
\end{equation}

\subsubsection{Fairness constraint on the consumer side}
We divide users into two groups of frequent users, $\mathcal{U}_1$, and less-frequent users, $\mathcal{U}_2$ based on their activity level. The notation $\mathcal{MC}$ is a metric that can evaluate the recommendation relevance in each group given the the recommendation list decision matrix, A, such as group $nDCG$, $F_1$ score, or $Recall$ on train or validation dataset.\footnote{In this paper we used $nDCG$ as our fairness metric calculated on training dataset for re-ranking.} The ideal fair recommendation would recommend items based on user's known preference regardless of their activity level and the quality of recommendation on frequent and less-frequent users should not get amplified by the recommendation algorithm. We define the deviation from consumer fairness parity as:

\begin{small}
\begin{equation}
\mathbf{DCF}(A, \mathcal{U}_1, \mathcal{U}_2) = \E\left[\mathcal{MC}(u,A) | u \in \mathcal{U}_1 \right] -\E\left[\mathcal{MC}(u,A) | u \in \mathcal{U}_2 \right]
    \label{eq:dcf}
\end{equation}
\end{small}

Note that when the values of $DPF\rightarrow 0$ or $DCF\rightarrow 0$, it shows parity between the advantaged and disadvantaged user and item groups. However, this may suggest a reduction in the overall system accuracy. The sign of $DCF$ and $DPF$ shows the direction of disparity among the two groups, which we decided to keep in this work.

\subsection{Proposed Algorithm}
In this section, we will use the top-N recommendation list and relevance scores provided by traditional unfair recommendation algorithm as baseline. We apply a re-ranking post-hoc algorithm to maximize the total sum of relevance scores while minimizing the deviation from fairness on the producer and consumer perspective, $\mathbf{DCF}$ and $\mathbf{DPF}$ respectively. We can formulate the fairness aware optimization problem given the binary decision matrix, A, as:

\begin{small}
\begin{equation}
\begin{aligned}
\max_{A_{ui}} \quad & \sum_{u \in \mathcal{U}}\sum_{i=1}^{N}{S_{ui}.A_{ui}} - \lambda_1. \mathbf{DCF}(A,\mathcal{U}_1,\mathcal{U}_2) - \lambda_2. \mathbf{DPF}(A,\mathcal{I}_1,\mathcal{I}_2)\\
\textrm{s.t.} \quad & \sum_{i=1}^{N}{A_{ui}} = K, A_{ui} \in{\{0,1\}}
\end{aligned}
\label{eq:optimization}
\end{equation}
\end{small}

The solution to above optimization problem would recommend exactly $K$ items to each user while maximizing the summation of preference scores and penalizing deviation from fairness. $ 0 \leq \lambda_1,\lambda_2 \leq 1$ are hyper parameters that would be selected based on the application and would assign how much weight is given to each fairness deviation concern. A complete user-oriented one-sided system could assign high weight to $\mathbf{DCF}$ and 0 to $\mathbf{DPF}$, and vice versa for a provider-oriented system. Note that if the value of $\lambda_1,\lambda_2$ are both equal to $0$ then the optimal solution will be the same as the baseline recommendation system list, that is, $L^{F}_K(u) = L_K(u)$. The effect of different values for $\lambda_1$ and $\lambda_2$ are studied in the ablation study section. Parity between groups is defined as the ideal representation of the fairness in the model but is not enforced as a constraint in our optimization to avoid resulting in diminishing the total model accuracy in a significant way. 

The optimization problem above is a mixed-integer linear programming (MILP) and a known NP-hard problem. Although MILP cannot be generally solved optimally in polynomial time, there exists many fast heuristic algorithms and industrial optimization solvers that would provide a good feasible solution in practice\footnote{\eg Gurobi solver at \url{https://www.gurobi.com}}. The MILP proposed optimization at Equation~\ref{eq:optimization} can be reduced to a special case of Knapsack problem where the goal is to select items for each user that would maximize the total score in which each item has a weight equal to 1 and the total weight is constrained by fixed size K. Since each item has the same weight in this instance of Knapsack problem, we can solve it optimally by a greedy algorithm in polynomial time. Now we are ready to propose the Fair Re-ranking Greedy algorithm (Algorithm~\ref{alg:LBP}) that would give the optimal solution for the optimization problem in Equation~\ref{eq:optimization}. Alternatively, we can relax the binary constraint of variables to fractional variables with boundary constraints \ie $0\leq A_{ui}\leq1$ and solve the relaxed optimization problem in polynomial time. The algorithm based on optimization is presented in Algorithm~\ref{alg:optmization}.

\begin{algorithm}
	\DontPrintSemicolon
	\SetAlgoLined
	\SetKwInOut{Input}{Input}\SetKwInOut{Output}{Output}
	\Input{\textit{$\mathcal{U},\mathcal{I}, \lambda_1,\lambda_2, K$}}
	\Output{Item recommendation matrix  {${A}^*$}}
	\tcc{${A}^*$ is a $n\times N$  binary matrix with elements $1$ when the item is recommended to the user and $0$ otherwise .}
	S $\leftarrow$ Run baseline algorithm and store the relevance score matrix for top-N items for each user \;
	\tcc{$S$ is a $n\times N$  matrix representing the predicted relevance score generated by unfair baseline algorithms for each user-item pair}
	Formulate the Optimization Problem as an Integer Programming following equation \ref{eq:optimization} \;
	Solve Relaxed Linear Programming with added boundary constraint $0\leq A_{ui}\leq1$, $\forall u\in \mathcal{U} , \forall i \in \{1,2,...,N \}$ \;
	\Return{$A^*$}
	\caption{The Fair Re-ranking Optimization}
	\label{alg:optmization}
\end{algorithm}

\begin{algorithm}
	\DontPrintSemicolon
	\SetAlgoLined
	\SetKwInOut{Input}{Input}\SetKwInOut{Output}{Output}
	\Input{\textit{$\mathcal{U},\mathcal{I}, \lambda_1,\lambda_2, K$}}
	\Output{Fair top-K recommendation list  {$L^{F}_K$}}
	
	S $\leftarrow$ Run baseline algorithm and store the relevance score matrix for top-N items for each user  \;
	$UG_{u}$ $\leftarrow$ Binary vector with elements of $1$ when user $u \in \mathcal{U}$ is in the protected group and $-1$ otherwise \;
	$PG_{i}$ $\leftarrow$ Binary vector with elements of $1$ when item ranked $i$ in the traditional top-N recommendation list, ${L}_N(u)$, is in the protected group and $-1$ otherwise \;

	\tcc{Protected groups are the underrepresented groups in baseline algorithms such as unpopular items }
	    
	\ForEach{$u\in \mathcal{U}$}{
        \ForEach{$i\in \{1,2,...,N\}$} {
            $CF_{ui} \leftarrow UG_{u} \times \mathcal{MC}(u,i)$ \;
            $PF_{ui} \leftarrow PG_{i} \times \mathcal{MP}(u,i)$ \;
		    $\hat{S}_{ui}\leftarrow S_{ui} + \lambda_1 CF_{ui} + \lambda_2 PF_{ui} $
		    
		    }
		}
   $L^{F}_K(u)$ $\leftarrow$ Sort matrix $\hat{S}$ along the rows axis and select top-$K$ items for recommendation to each users, $\forall$ $u\in \mathcal{U}$ \;
	\Return{$L^{F}_K(u)$}
	\caption{The Fair Re-ranking Greedy Algorithm}
	\label{alg:LBP}
\end{algorithm}

We finally discuss the time complexity of the Fair Re-ranking Greedy algorithm.

\begin{lemma}
The time complexity of proposed re-ranking algorithm has a worst case bound of $\mathcal{O}(n\times N)$
\end{lemma}

\begin{proof}
The first 3 lines of the algorithm will store the traditional relevance score matrix for Top-N items and store the protected groups in consumers and providers which can be done in linear time $\mathcal{O}(n + m)$. The second part in line 4-10 is to re-score all elements in the original relevance matrix by accounting for fairness measures which needs $\mathcal{O}(n \times N)$ operations in constant time. Lastly, we will sort the new relevance matrix for all users and select top-K items for recommendation. Using max-heap the worst case time complexity is $\mathcal{O}(N + K \log(K))$. Hence, considering $N$ and $K$ are much smaller than $n$, the total time complexity is $\mathcal{O}(n \times N)$ 
\end{proof}
\section{Experimental Methodology}
\label{sec:experimental_methodology}
In this section, we first briefly describe the datasets, baselines and experimental setup used for experiments. All source code and dataset of this project has been released publicly. To foster the reproducibility of our study, we evaluate all recommendations with the Python-based open-source recommendation toolkit Cornac\footnote{\url{https://cornac.preferred.ai/}} \cite{salah2020cornac,truong2021exploring}, and we have made our codes open source.\footnote{\url{https://rahmanidashti.github.io/CPFairRecSys/}}

\partitle{Datasets.}
We use eight public datasets from different domains such as Movie, Point-of-Interest, Music, Book, E-commerce, and Opinion with different types of feedback (\ie explicit or implicit) and characteristics, including \datasetA, \datasetB, \datasetC, and \datasetD. We apply $k$-core pre-processing (\ie each training user/item has at least $k$ ratings) on the datasets to make sure each user/item has sufficient feedback. Table \ref{tbl:datasets} shows the statistics of the all datasets.

\begin{table}
  \caption{Statistics of the datasets}
  \centering
  \label{tbl:datasets}
  \begin{tabular}{lcccc}
    \toprule
    \textbf{Dataset} & Users & Items & Interactions & Sparsity \\
    \midrule
    \textbf{\datasetA}  & 943 & 1,349 & 99,287 & 92.19\% \\
    \textbf{\datasetB}  & 2,677 & 2,060 & 103,567 & 98.12\% \\
    \textbf{\datasetC}  & 1,130 & 1,189 & 66,245 & 95.06\% \\
    \textbf{\datasetD}  & 1,797 & 1,507 & 62,376 & 97.69\% \\
    \textbf{BookCrossing}  & 1,136 & 1,019 & 20,522 & 98.22 \% \\
    \textbf{Foursquare}  & 1,568 & 1,461 & 42,678 & 98.13\% \\
    \textbf{AmazonToy}  & 2,170 & 1,733 & 32,852 & 99.12\% \\
    \textbf{AmazonOffice}  & 2,448 & 1,596 & 36,841 & 99.05\% \\
    \bottomrule 
  \end{tabular}
\end{table}

\partitle{Evaluation Method.}
For the experiments, we randomly split each dataset into train, validation, and test sets in 70\%, 10\%, and 20\% respectively, and all the baseline models share these datasets for training and evaluation. For the re-ranking experiment, we categorize both users and items into two groups \cite{li2021user,abdollahpouri2019unfairness}. For users, we select the top 5\% users based on the number of interaction/level of activity from the training set as the active group and the rest as the inactive group. For items, we select top 20\% of items according to the number of interactions they received from the training set as the popular items, \ie short-head items, and the rest as the unpopular items or long-tail items. In addition, we define Consumer-Producer fairness evaluation (mCPF) metric to evaluate the overall performance of models \wrt two-sided fairness. mCPF computes the weighted average deviation of provider fairness (DPF) and deviation of consumer fairness (DCF), \ie $mCPF(w) = w \cdot DPF + (1-w) \cdot DCF$, where $w$ is a weighting parameter that can be selected by the system designer depending on the application or priority for each fairness aspect of the \textit{marketplace}. DCF and DPF are defined in Equations \ref{eq:dpf} and \ref{eq:dcf}, respectively, and the lowest value for mCPF shows the fairest model. To select the model parameters, $\lambda_1$ and $\lambda_2$, we treat them as a hyper parameter and set them to the value that maximize the average values of mCPF and overall nDCG for each dataset and model.

\partitle{Baselines.}
We compare the performance of our fairness method on several recommendation approaches, from traditional to deep recommendation models, as suggested by \citet{dacrema2019we}. Therefore, we include two traditional methods (PF and WMF) as well as two deep recommendation models (NeuMF and VAECF). We also include two baselines approaches, Random and MostPop, for further investigation of the results. The introduction of baselines are as the following:  

\begin{itemize}
    \item \textbf{PF} \cite{gopalan2015scalable}: This method is a variant of probabilistic matrix factorization where the weights of each user and item latent features are positive and modeled using the Poisson distribution.
    \item \textbf{WMF} \cite{hu2008collaborative,pan2008one}: This method assigns smaller wights to negative samples and assumes that for two items their latent features are independent.
    \item \textbf{NeuMF} \cite{he2017neural}: This algorithm learns user and item features using multi-layer perceptron (MLP) on one side and matrix factorization (MF) from another side, then applies non-linear activation functions to train the mapping between users and items features that are concatenated from MLP and MF layers.
    \item \textbf{VAECF} \cite{liang2018variational}: This method is based on variational autoencoders which introduces a generative model with multinomial likelihood and use Bayesian inference for parameter estimation.
\end{itemize}

\noindent We compare our proposed fair CP re-ranking algorithm with unilateral fairness models, C and P-fairness re-ranking on top of the baselines to show how our method can achieve the desirable performance on fairness metrics in both stakeholders and overall recommendation performance (accuracy and beyond-accuracy).

\partitle{Evaluation Settings.}
We adopt the baseline algorithms with the default parameter settings, suggested in their original paper. We set the embedding size for users and items to $50$ for all baseline algorithms. For NeuMF, we set the size of the MLP with $32, 16, 8$ and we apply hyperbolic tangent (TanH) non-linear activation function between layers. We set the learning rate to $0.001$. Early stopping is applied and the best models are selected based on the performance on the validation set. We apply Adam \cite{kingma2014adam} as the optimization algorithm to update the model parameters.
\section{Results}
\label{sec:results}

\begin{table*}
\centering
\caption{The recommendation performance of all, and different user and item groups of our re-ranking method and corresponding baselines on \datasetB and \datasetC datasets. All re-ranking results here are obtained under the fairness constraint on nDCG. The evaluation metrics here are calculated based on the top-10 predictions in the test set. Our best results are highlighted in bold. $\Delta$ values denote the percentage of relative improvement in $mCPF$ compared to fairness-unaware recommendation algorithm (N). The reported results for $mCPF$ are based on $w=0.5$.}
\label{tbl:results_epinion}
\begin{tabular}{llllllllllllllll}
\toprule
\multirow{2}{*}{Model} & \multirow{2}{*}{Type} & \multicolumn{4}{c}{User Relevance (nDCG)} && \multicolumn{5}{c}{Item Exposure} && \multicolumn{3}{c}{Both} \\
\cmidrule{3-6}\cmidrule{8-12}\cmidrule{14-16}
                        &                       & All & Active & Inactive & DCF $\downarrow$ && Nov.  & Cov.  & Short. & Long. & DPF $\downarrow$ && mCPF $\downarrow$ & $\frac{mCPF}{All}$ $\downarrow$ & $\Delta(\%)$\\
\midrule
\multicolumn{16}{c}{\textbf{\datasetB}} \\ \hline
PF & N & 0.0321 & 0.0751 & 0.0298 & 0.4321 && 4.9602 & 0.5073 & 0.8281 & 0.1719 & 0.6562 && 0.621 & 19.3458 & 0.0 \\ 
PF & C & \textbf{0.0337} & 0.0732 & \textbf{0.0316} & \textbf{0.3964} && 4.9754 & 0.5083 & 0.8243 & 0.1757 & 0.6486 && 0.593 & 17.5964 & 4.50 \\ 
PF & P & 0.0304 & \textbf{0.0757} & 0.028 & 0.4598 && 5.3172 & 0.5325 & 0.6153 & 0.3847 & 0.2306 && 0.427 & 14.0461 & 31.23 \\ 
PF & CP & 0.0313 & 0.074 & 0.029 & 0.4368 && \textbf{5.3302} & \textbf{0.5345} & \textbf{0.6113} & \textbf{0.3887} & \textbf{0.2226} && \textbf{0.4074} & \textbf{13.0160} & \textbf{34.39} \\ \hline
WMF & N & 0.0235 & 0.0577 & 0.0217 & 0.4534 && 4.9483 & 0.2913 & 0.9423 & 0.0577 & 0.8846 && 0.7496 & 31.8979 & 0.0 \\ 
WMF & C & \textbf{0.0236} & 0.0565 & \textbf{0.0219} & \textbf{0.442} && 4.9481 & 0.2898 & 0.9423 & 0.0577 & 0.8846 && 0.7419 & 31.4364 & 1.02 \\ 
WMF & P & 0.0234 & \textbf{0.0577} & 0.0215 & 0.456 && \textbf{4.9626} & \textbf{0.3019} & 0.9322 & 0.0678 & 0.8644 && 0.7413 & 31.6795 & 1.10 \\ 
WMF & CP & 0.0234 & 0.0565 & 0.0217 & 0.445 && 4.9623 & 0.301 & \textbf{0.9322} & \textbf{0.0678} & \textbf{0.8644} && \textbf{0.7339} & \textbf{31.3632} & \textbf{2.09} \\ \hline
NeuMF & N & 0.0447 & 0.084 & 0.0427 & 0.3264 && 3.8733 & 0.1223 & 0.9914 & 0.0086 & 0.9828 && 0.7127 & 15.9441 & 0.0 \\ 
NeuMF & C & \textbf{0.0456} & 0.077 & \textbf{0.044} & 0.2726 && 3.8745 & 0.1223 & 0.9913 & 0.0087 & 0.9826 && 0.6761 & 14.8268 & 5.13 \\ 
NeuMF & P & 0.0448 & \textbf{0.0846} & 0.0427 & 0.3297 && 3.8801 & 0.1262 & \textbf{0.9849} & \textbf{0.0151} & \textbf{0.9698} && 0.7084 & 15.8125 & 0.60 \\ 
NeuMF & CP & 0.0455 & 0.0749 & 0.0439 & \textbf{0.2606} && \textbf{3.8816} & \textbf{0.1262} & 0.9851 & 0.0149 & 0.9702 && \textbf{0.6618} & \textbf{14.5451} & \textbf{7.14} \\ \hline
VAECF & N & \textbf{0.0444} & \textbf{0.0801} & 0.0425 & 0.3067 && 4.4061 & 0.3107 & 0.919 & 0.081 & 0.838 && 0.6269 & 14.1194 & 0.0 \\
VAECF & C & 0.0443 & 0.0735 & \textbf{0.0428} & \textbf{0.2644} && 4.408 & 0.3102 & 0.9184 & 0.0816 & 0.8368 && 0.5976 & 13.4898 & 4.67 \\
VAECF & P & 0.0438 & 0.0789 & 0.042 & 0.3054 && 4.4349 & \textbf{0.3218} & 0.8946 & 0.1054 & 0.7892 && 0.6016 & 13.7352 & 4.03 \\
VAECF & CP & 0.0439 & 0.0739 & 0.0423 & 0.2716 && \textbf{4.4366} & 0.3204 & \textbf{0.8943} & \textbf{0.1057} & \textbf{0.7886} && \textbf{0.5784} & \textbf{13.1754} &\textbf{7.73} \\
\midrule
\multicolumn{16}{c}{\textbf{\datasetC}} \\ \hline
PF & N & 0.0592 & 0.0685 & 0.0587 & 0.0771 && 4.0587 & 0.6098 & 0.7026 & 0.2974 & 0.4052 && 0.2569 & 4.3395 & 0.0 \\ 
PF & C & \textbf{0.0662} & 0.0617 & \textbf{0.0664} & 0.0367 && 4.0701 & 0.614 & 0.7011 & 0.2989 & 0.4022 && 0.2269 & 3.4275 & 11.67 \\ 
PF & P & 0.057 & \textbf{0.0688} & 0.0564 & 0.0993 && 4.3348 & 0.6468 & 0.4954 & 0.5046 & -0.0092 && 0.0653 & 1.1456 & 74.58 \\ 
PF & CP & 0.0618 & 0.0585 & 0.062 & \textbf{0.0288} && \textbf{4.3559} & \textbf{0.651} & \textbf{0.488} & \textbf{0.512} & \textbf{-0.024} && \textbf{0.0083} & \textbf{0.1343} & \textbf{96.76} \\ \hline
WMF & N & 0.0338 & \textbf{0.055} & 0.0327 & 0.2541 && 4.5012 & 0.4617 & 0.5775 & 0.4225 & 0.155 && 0.2564 & 7.5858 & 0.0 \\ 
WMF & C & \textbf{0.0357} & 0.0527 & \textbf{0.0348} & 0.2044 && 4.5152 & 0.4567 & 0.5748 & 0.4252 & 0.1496 && 0.2187 & 6.1261 & 14.70 \\ 
WMF & P & 0.0318 & 0.0508 & 0.0308 & 0.2454 && 4.649 & \textbf{0.4693} & 0.4757 & 0.5243 & -0.0486 && 0.1485 & 4.6698 & 42.08 \\ 
WMF & CP & 0.0334 & 0.047 & 0.0327 & \textbf{0.1793} && \textbf{4.6657} & 0.4626 & \textbf{0.4705} & \textbf{0.5295} & \textbf{-0.059} && \textbf{0.0967} & \textbf{2.8952} & \textbf{62.28} \\ \hline
NeuMF & N & 0.0563 & \textbf{0.1061} & 0.0537 & 0.3283 && 3.3409 & 0.1463 & 0.9245 & 0.0755 & 0.849 && 0.6556 & 11.6448 & 0.0 \\ 
NeuMF & C & \textbf{0.0585} & 0.0897 & \textbf{0.0568} & 0.2244 && 3.3427 & 0.1447 & 0.9242 & 0.0758 & 0.8484 && 0.5822 & 9.9521 & 11.19 \\ 
NeuMF & P & 0.0552 & 0.1049 & 0.0526 & 0.332 && \textbf{3.3827} & \textbf{0.1531} & 0.8868 & 0.1132 & 0.7736 && 0.6205 & 11.2409 & 5.35 \\ 
NeuMF & CP & 0.0583 & 0.0889 & 0.0566 & \textbf{0.2214} && 3.3823 & 0.1522 & \textbf{0.8868} & \textbf{0.1132} & \textbf{0.7736} && \textbf{0.5427} & \textbf{9.3087} & \textbf{17.22} \\ \hline
VAECF & N & 0.0608 & 0.093 & 0.0592 & 0.2227 && 3.8277 & 0.4104 & 0.7106 & 0.2894 & 0.4212 && 0.3674 & 6.0428 & 0.0 \\
VAECF & C & \textbf{0.0623} & 0.0916 & \textbf{0.0607} & 0.2028 && 3.8289 & 0.4113 & 0.7104 & 0.2896 & 0.4208 && 0.3532 & 5.6693 & 3.86 \\
VAECF & P & 0.06 & \textbf{0.0935} & 0.0582 & 0.2324 && 3.8492 & 0.4146 & 0.6924 & 0.3076 & 0.3848 && 0.356 & 5.9333 & 3.10 \\
VAECF & CP & 0.0619 & 0.0889 & 0.0605 & \textbf{0.19} && \textbf{3.8501} & \textbf{0.4146} & \textbf{0.6919} & \textbf{0.3081} & \textbf{0.3838} && \textbf{0.3257} & \textbf{5.2617} & \textbf{11.35} \\
\bottomrule
\end{tabular}
\end{table*}

\begin{figure*}
  \centering
  \subfloat[MovieLens]
    {\includegraphics[scale=0.27]{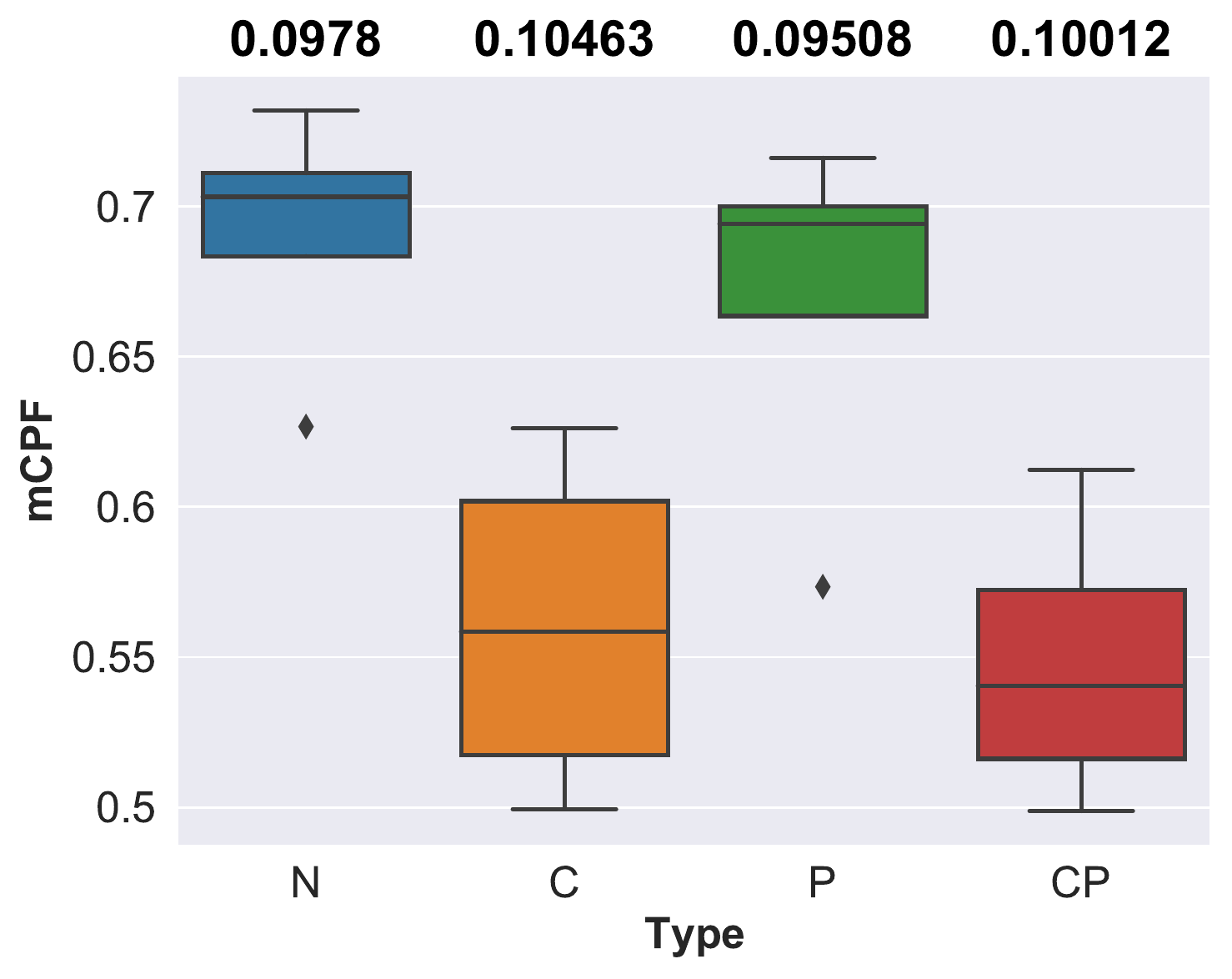}
    \label{fig:CPevalBoxPlot_MovieLens}}
  \hfill
  \subfloat[Epinion]
    {\includegraphics[scale=0.27]{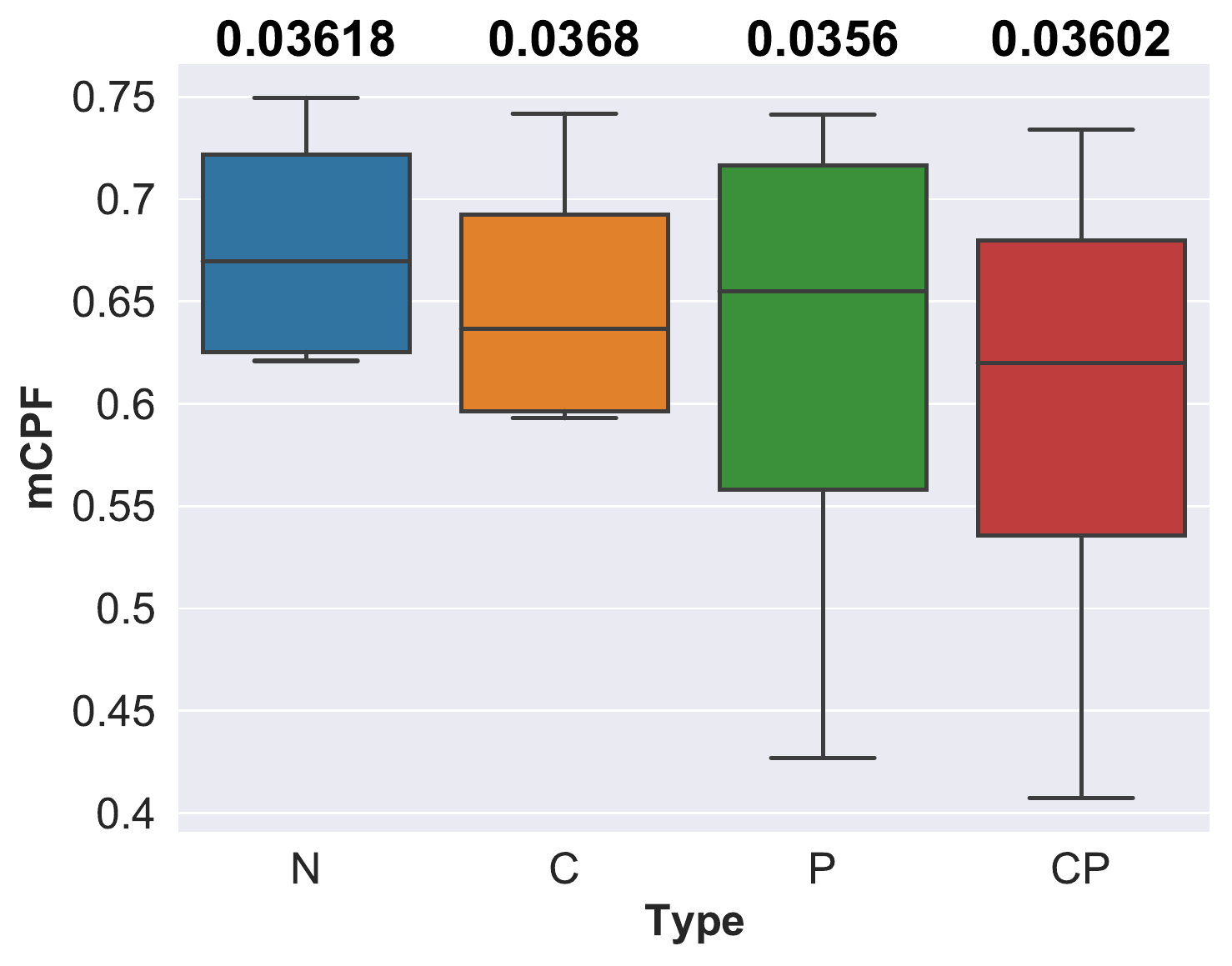}
    \label{fig:CPevalBoxPlot_Epinion}}
  \hfill
  \subfloat[AmazonToy]
    {\includegraphics[scale=0.27]{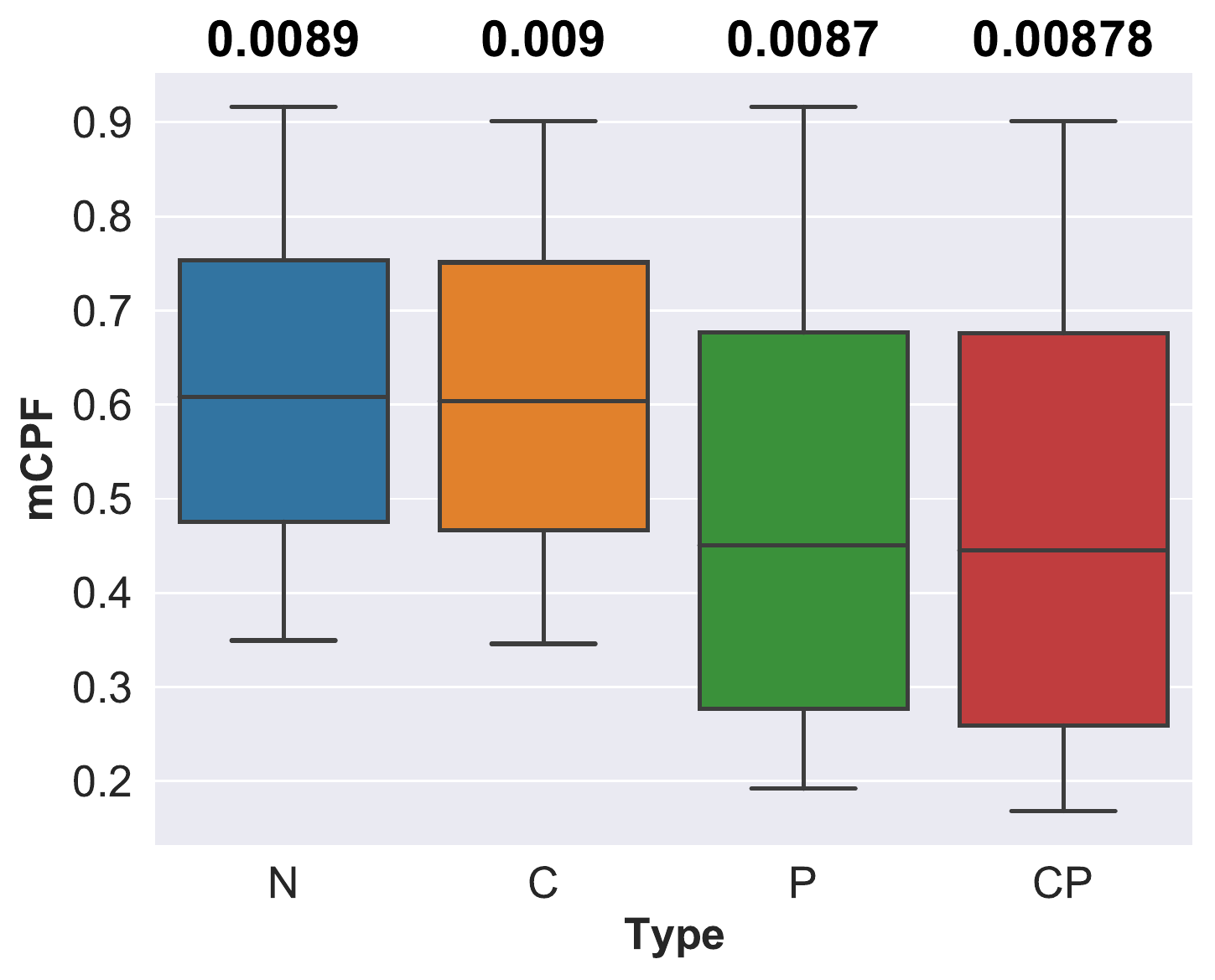}
    \label{fig:CPevalBoxPlot_AmazonToy}}
    \hfill
  \subfloat[AmazonOffice]
    {\includegraphics[scale=0.27]{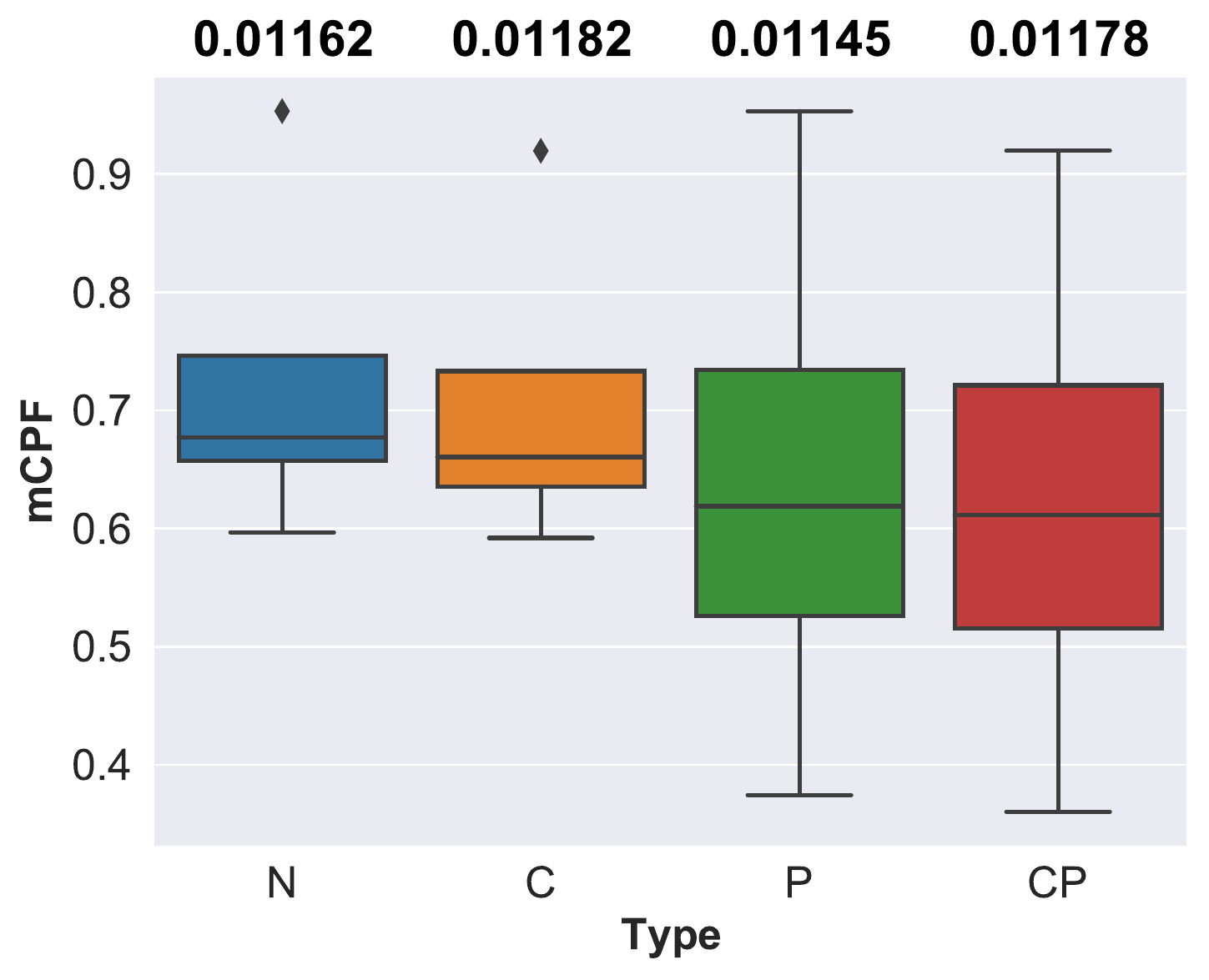}
    \label{fig:CPevalBoxPlot_AmazonOffice}}
   \hfill
  \subfloat[BookCrossing]
    {\includegraphics[scale=0.27]{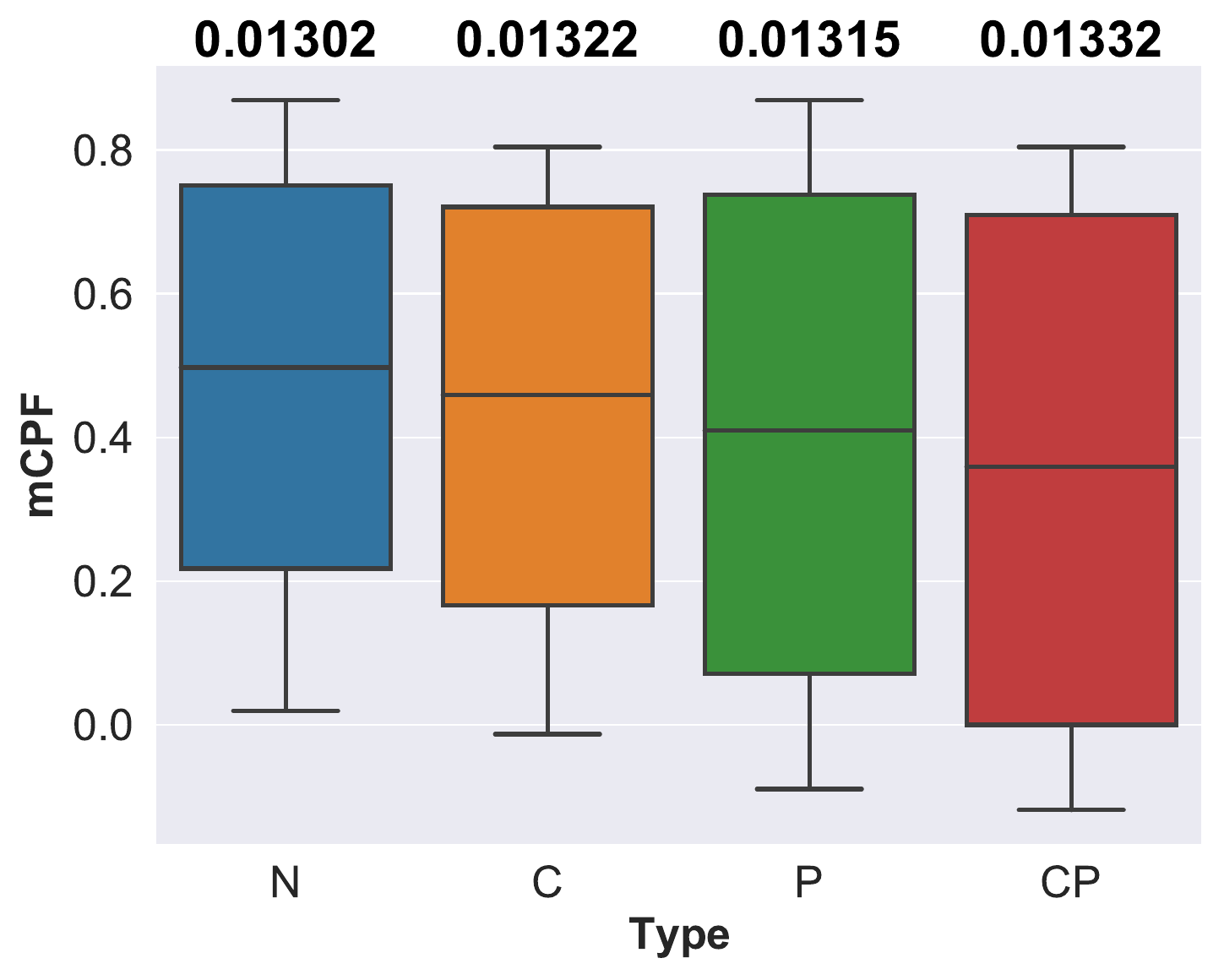}
    \label{fig:CPevalBoxPlot_BookCrossing}}
  \hfill
  \subfloat[Gowalla]
    {\includegraphics[scale=0.27]{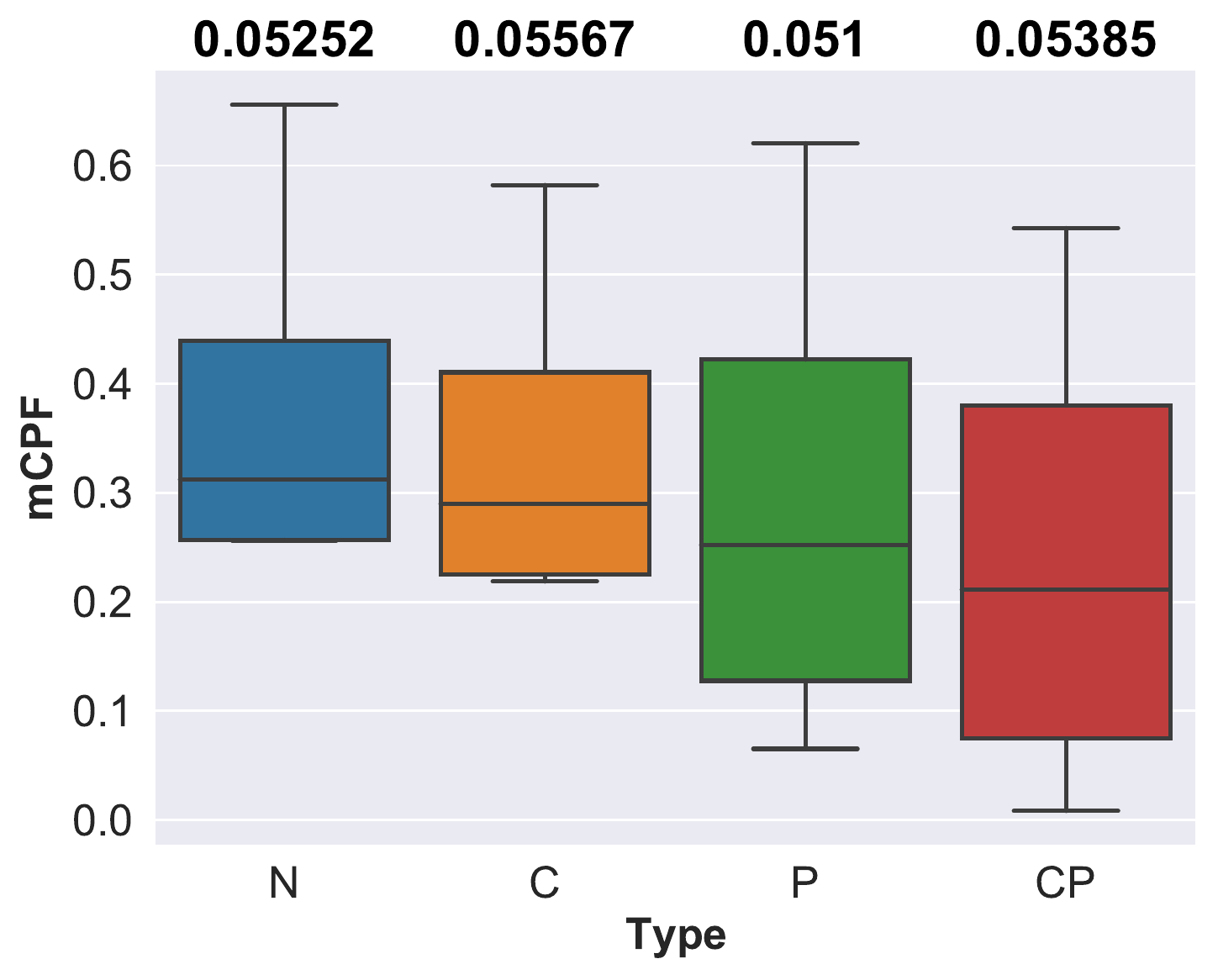}
    \label{fig:CPevalBoxPlot_Gowalla}}
  \hfill
  \subfloat[LastFM]
    {\includegraphics[scale=0.27]{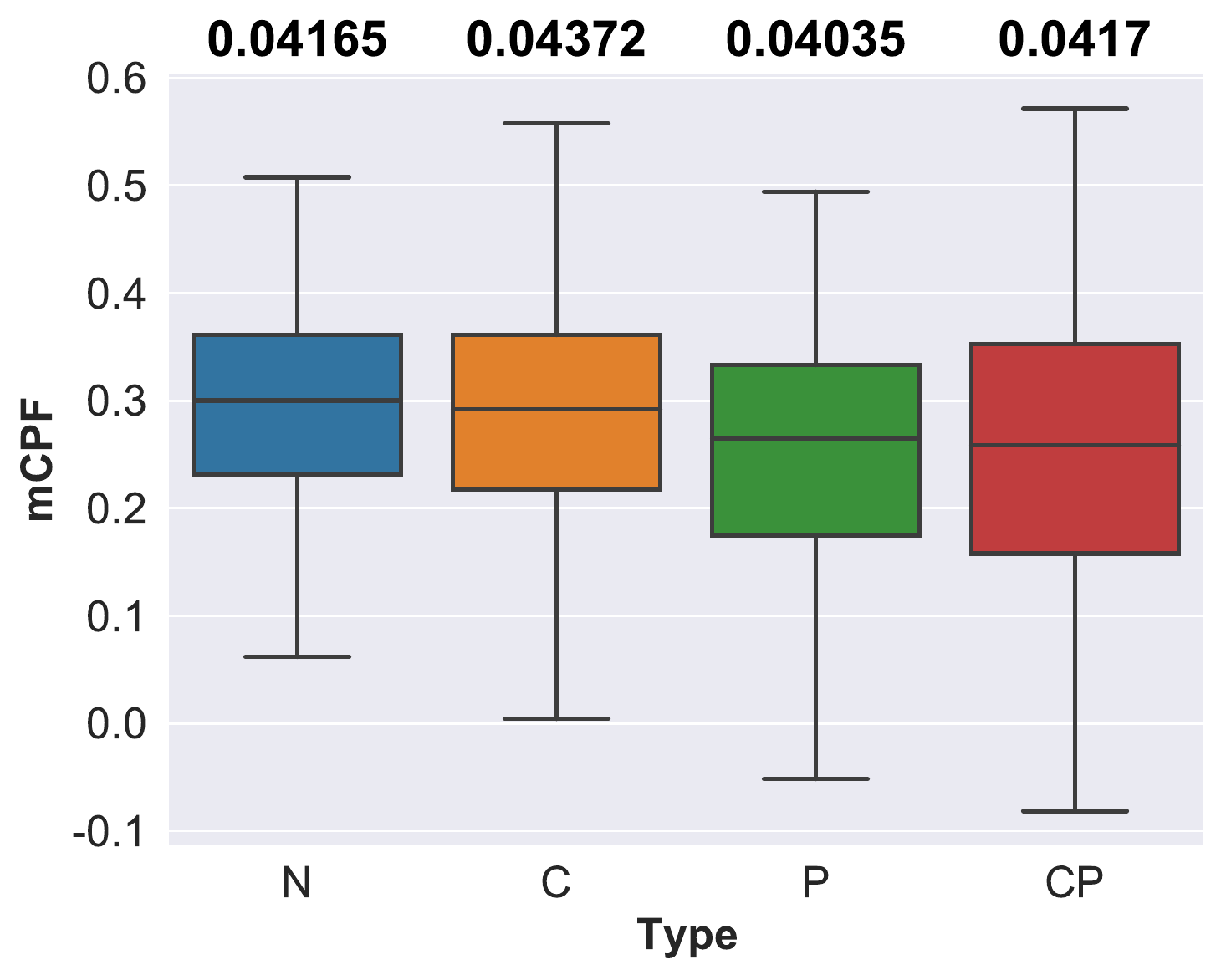}
    \label{fig:CPevalBoxPlot_LastFM}}
  \hfill
  \subfloat[Foursquare]
    {\includegraphics[scale=0.27]{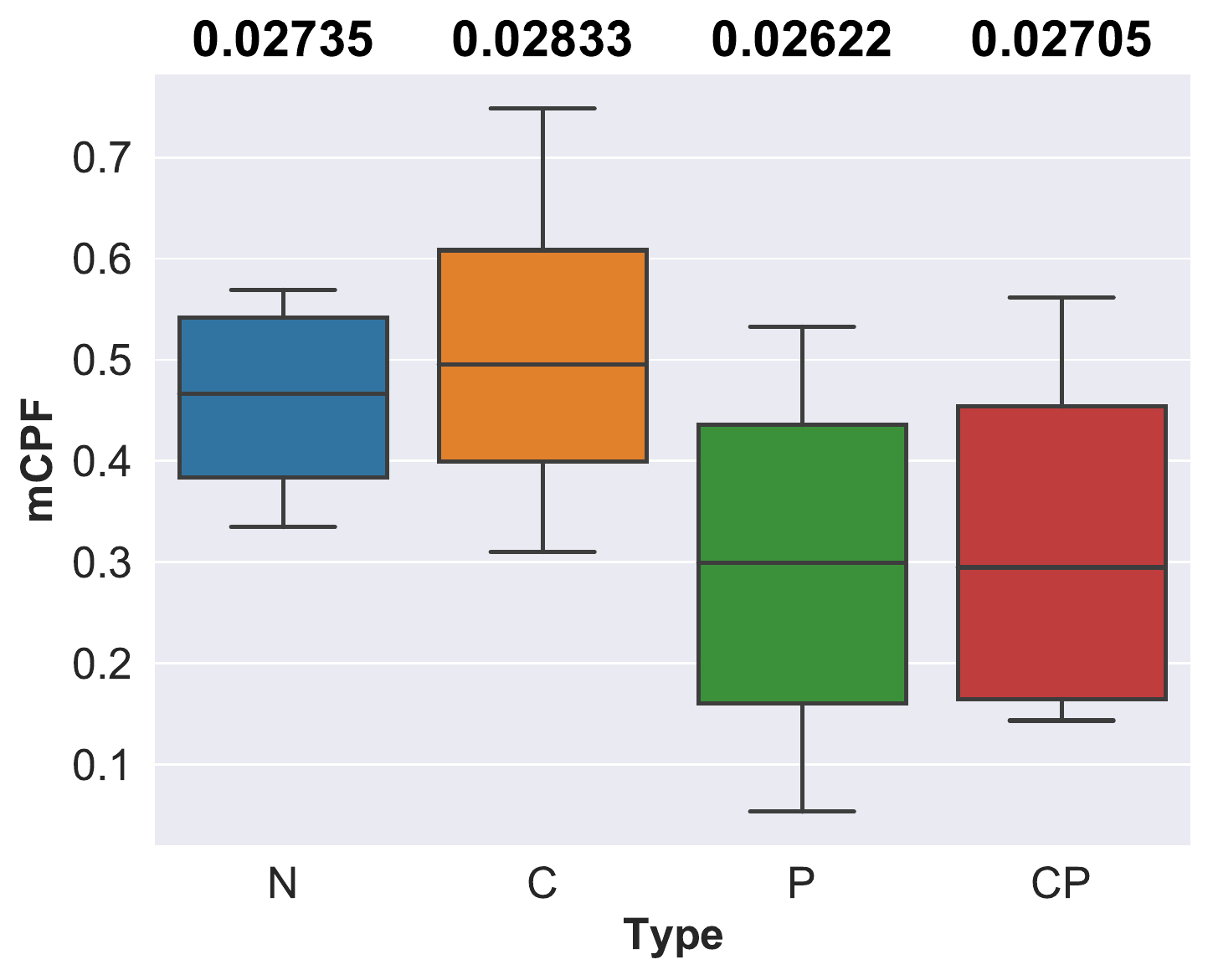}
    \label{fig:CPevalBoxPlot_Fourquare}}
  \caption{Distributions of mCPF for fairness-unaware (N) and fairness-aware methods (i.e., C, P, and CP) on all 8 datasets. The numbers on top of each plot show the overall performance (i.e., nDCG@10 for all users) according to each fairness methodology.}
\label{fig:CPevalBoxPlot}
\end{figure*}

\begin{figure*}
  \centering
  \subfloat[User Groups]
    {\includegraphics[scale=0.27]{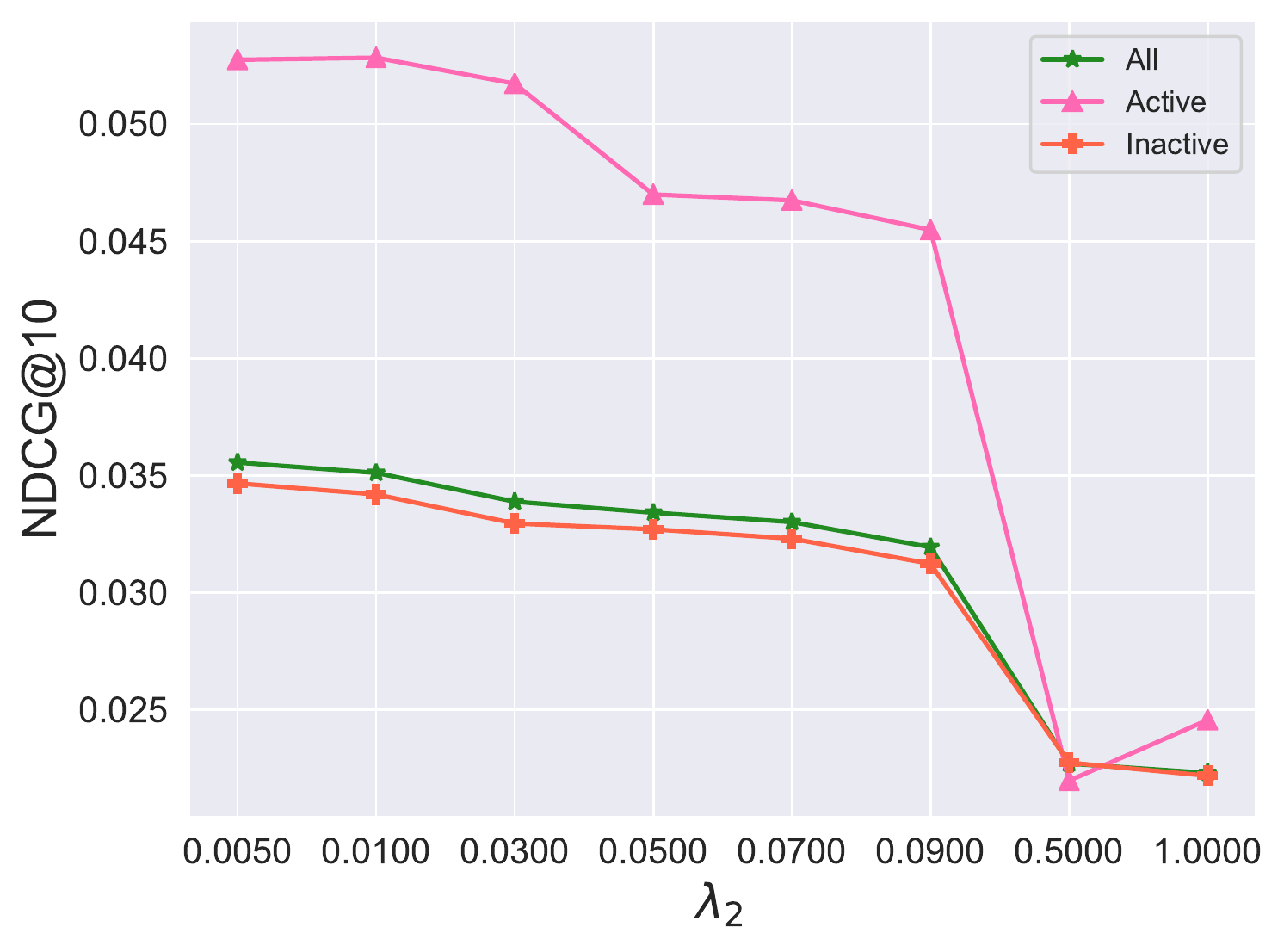}
    \label{fig:ablation_user_ueps}}
  \hfill
  \subfloat[User Groups]
    {\includegraphics[scale=0.27]{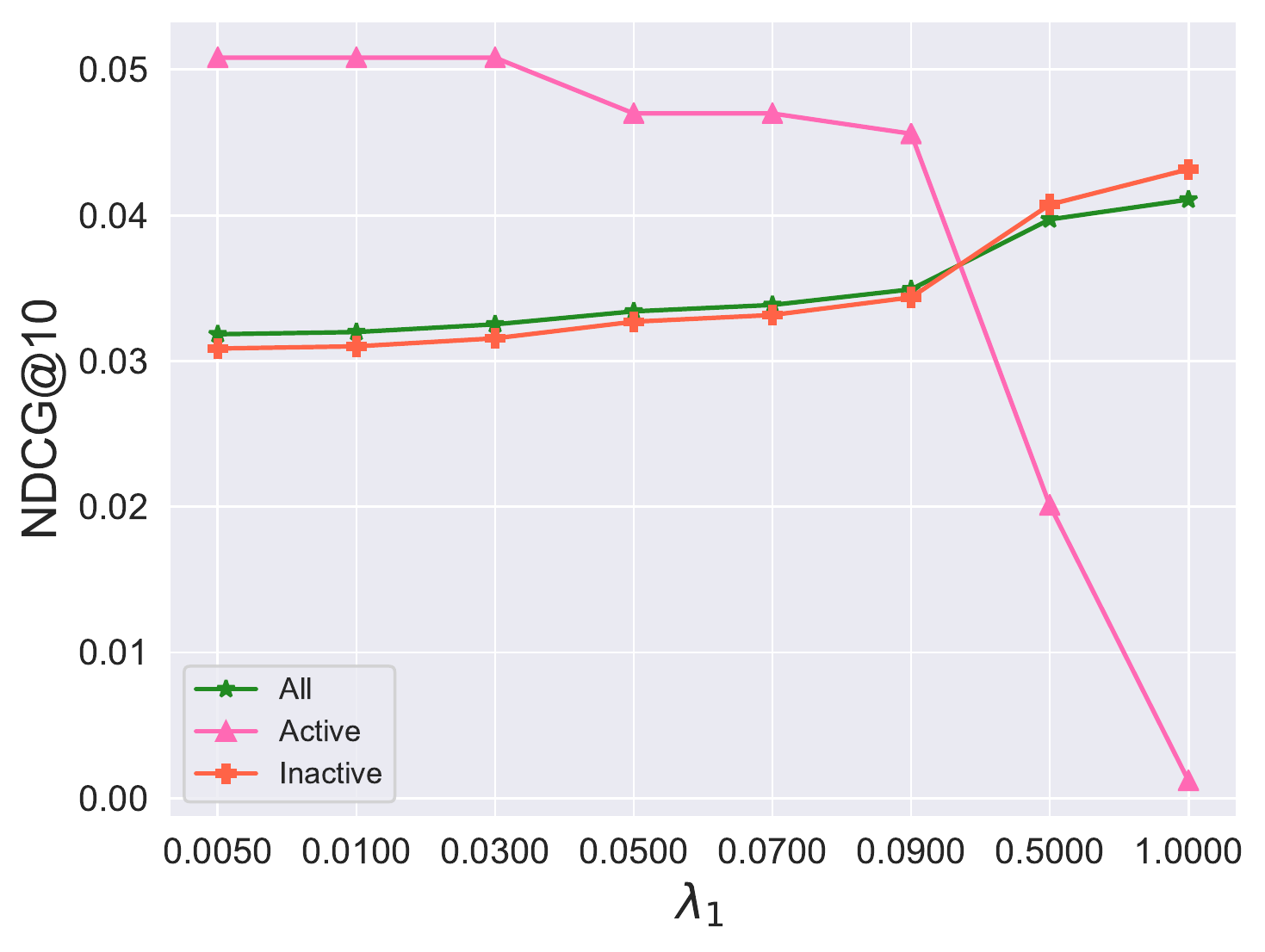}
    \label{fig:ablation_user_ieps}}
    \hfill
  \subfloat[Item Groups]
    {\includegraphics[scale=0.27]{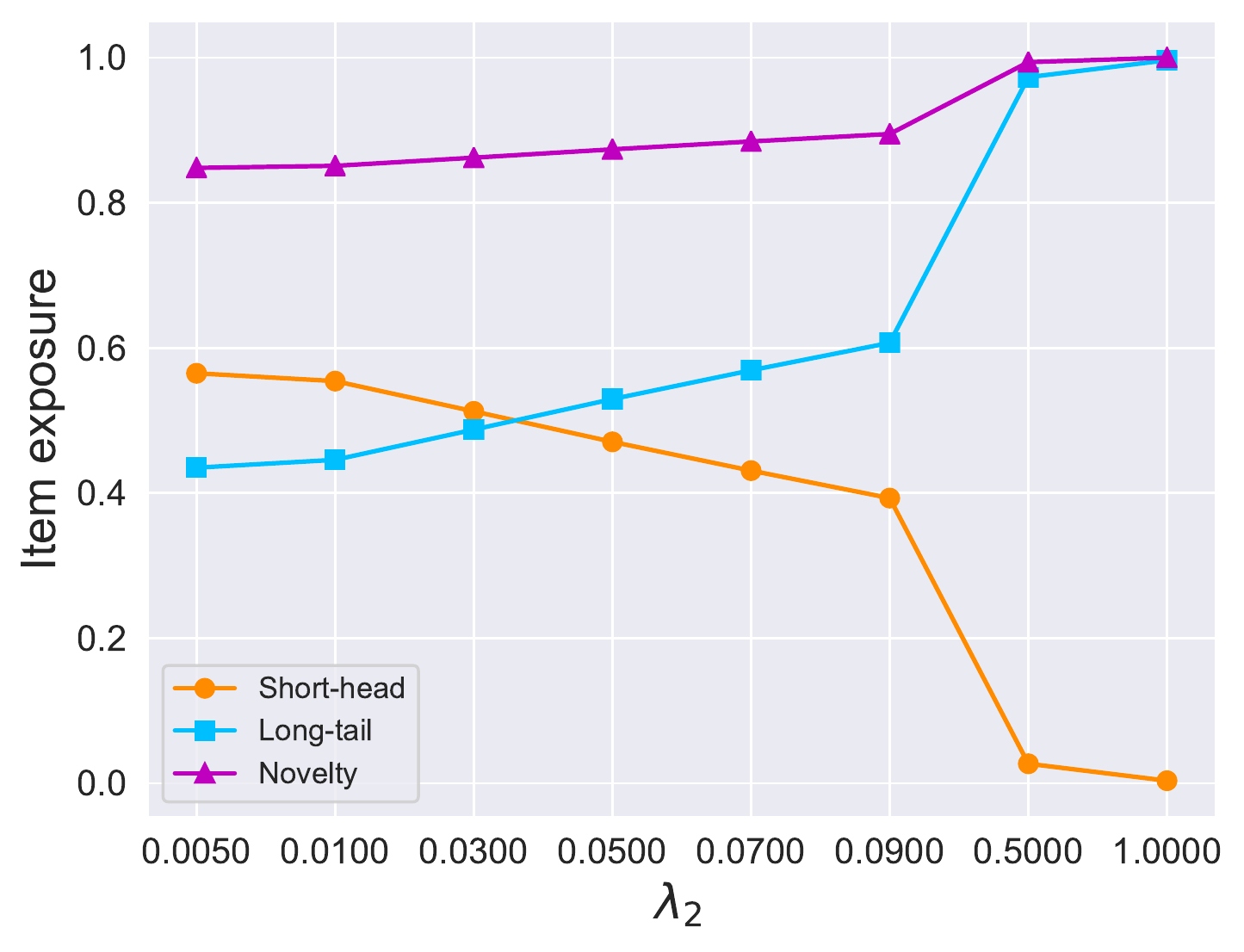}
    \label{fig:ablation_item_ueps}}
   \hfill
  \subfloat[Item Groups]
    {\includegraphics[scale=0.27]{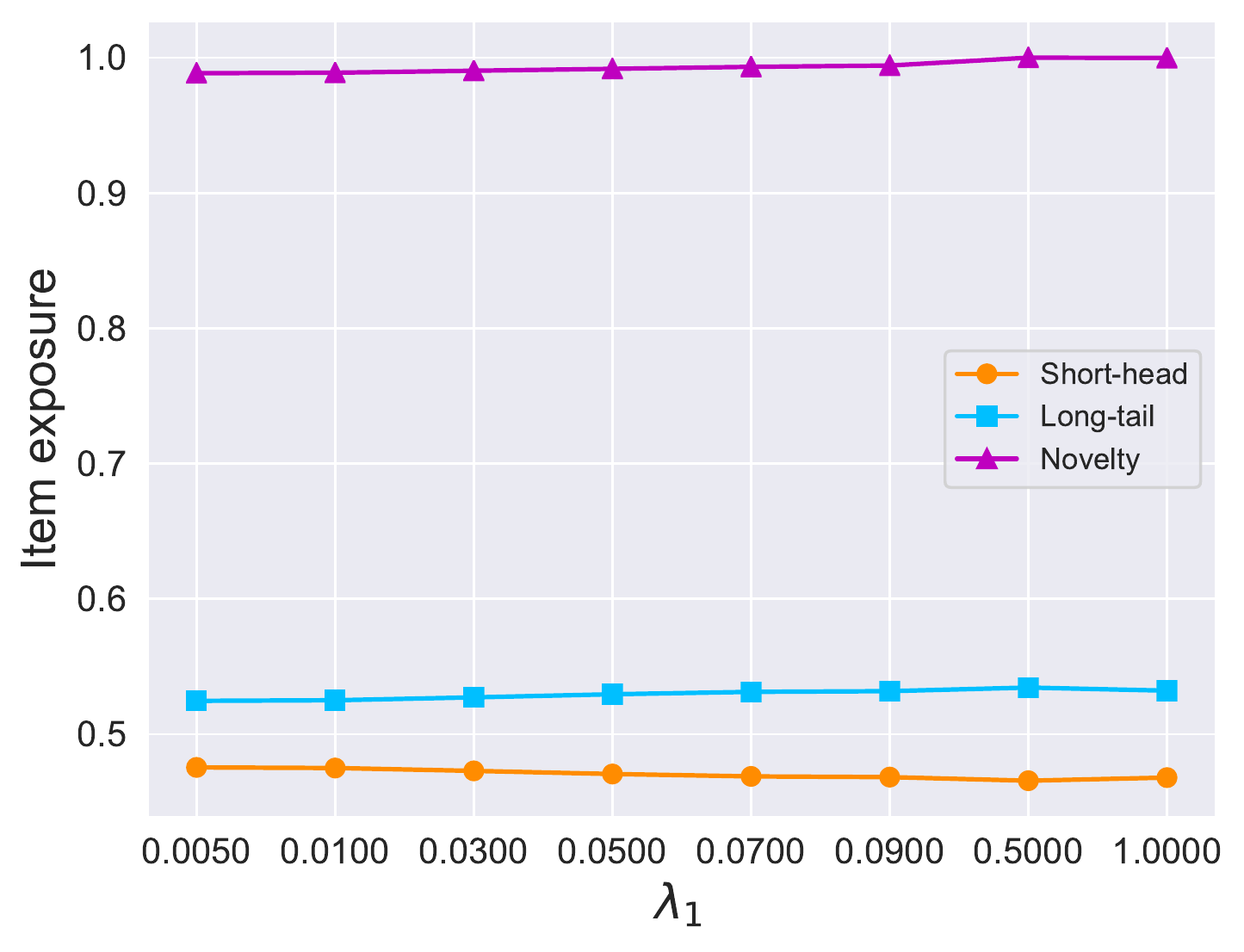}
    \label{fig:ablation_item_ieps}}
  \caption{The metric nDCG@10 and item exposure change on CP-fairness with respect to the $\lambda_1$ and $\lambda_2$ on all, \usergroupA and \usergroupB user groups and \itemgroupA and \itemgroupB item group. In each figure, the other $\lambda$ is equal to $0.05$.}
\label{fig:ablationGowallaWMF}
\end{figure*}

To gain a better understanding of the merits of the proposed CP fairness optimization framework, through the course of experiments, we intend to answer the following evaluation questions:
\begin{description}
    \item \textbf{Q1}: When comparing different \underline{fairness-unaware} baseline CF recommendation models, is there an underlying relationship between the fairness constraints of user-relevance, fairness in item exposure, and the overall system accuracy?
    \item \textbf{Q2}: How does \underline{one-sided} fairness optimization (C-fairness~\cite{li2021user,islam2021debiasing} or P-Fairness~\cite{gomez2022provider,boratto2021interplay}), which has been the main subject of \underline{previous research}, affect the fairness constraints of other stakeholders, for example, if maximizing P-Fairness (i.e., reducing DPF) results in a loss in the user-fairness and the total system accuracy?
    \item \textbf{Q3}: Does the \underline{proposed CP-fairness} algorithm in this work meet the requirements for fairness on both the consumer and supplier sides simultaneously, and whether this comes at a cost in terms of overall system accuracy or not necessarily? 
\end{description}
 
\partitle{Results and Discussion.} We begin our experimental study by addressing the above evaluation questions. We report the detailed performance results of two datasets \datasetB and \datasetC (out of eight) in Table~\ref{tbl:results_epinion}, due to space considerations. However, using the boxplots provided in Figure~\ref{fig:CPevalBoxPlot}, we are able to represent the aggregate average findings across all of the datasets in a graphical format. 

\partitle{\textit{Answer to Q1.}}
The answer to this question may be validated in Section~\ref{sec:motivating_cp} and by examining the column type \squotes{N} in Table~\ref{tbl:results_epinion}. We calculate the aggregate average values over four CF models (PF, WMF, NeuMF and VACEF) in the \squotes{N} category to draw the following findings. The mean value on the columns (Overall Acc., DCF, DPF) are equal to, for \datasetB=$(0.0362, 0.3797, 0.8404)$ and \datasetC=$(0.0525, 0.2205, 0.4576)$. In addition, to find the best-performing methods, we use the normalized user-item biases over accuracy, $mDCF/All$, which for the four CF models on the  \datasetB dataset and \datasetC datasets respectively equals to: ($19.3458$, $31.8979$, $15.9441$, $14.1194$) and ($4.3395$, $7.5858$, $11.6448$, $6.0428$). Note that the lower these values, the more capable the model is is in balancing CP-fairness constraints and the overall system accuracy.

\partitle{\underline{\textbf{Observation.}}}
By examining the aggregated average results, we can clearly note that the data contain inherent biases and imbalances that negatively impact consumers and suppliers. The bias in \datasetB is substantially worse than that in \datasetC; for example, compare the differences in mean DCF: (0.3797, 0.2205) and mean DPF: (0.8404, 0.4576) across the four baseline CF models. The differences in values may be answered based on the underlying data characteristics~\cite{deldjoo2021explaining}.

Considering the trade-off between accuracies and CP-fairness, the first observation is that some algorithms (e.g., NeuMF) are more prone to amplifying the bias, and some less. For instance, in \datasetB, NeuMF achieves the best accuracy (0.0447), however, by hugely sacrificing CP-fairness or mCPF (from 0.621 to 0.7127). When \textbf{relative fairness-accuracy} is concerned, according to the $mDCF/All$ column, NeuMF still achieves the second-best performance $mCPF/All = 15.94$ in \datasetB. However, in \datasetC, this methods has the \underline{worst} performance with $mCPF/All = 11.64$. These results imply that the performance of some CF models (as in this case NeuMF) varies substantially as the characteristics of the underlying dataset are altered, while some (such as VACEF) are more robust to these variations. For instance, the VACEF technique outperforms all other methods in establishing the optimal balance of CP-Fairness and accuracies, placing second in \datasetB and first in \datasetC. These findings and insights underscore the critical importance of algorithmic consumer and supplier fairness in preserving a two-sided eco-system, which is the focus of the current research.

\partitle{\textit{Answer to Q2.}}
The aggregate mean over four CF baseline algorithms (PF, WMF, NeuMF and VACEF) in one-sided fairness types, C and P on columns (DCF, DPF, mCPF) is respectively equal to: ($0.3439$, $0.8382$, $0.6522$) and ($0.3877$, $0.7135$, $0.6196$), in \datasetB and (0.16708, 0.45525, 0.34525) and (0.22728, 0.27515, 0.29758) in \datasetC. The baseline N on the two datasets are equal to (0.3797, 0.8404, 0.6776) and (0.22055, 0.45760, 0.38408) respectively. 

\partitle{\underline{\textbf{Observation.}}}
These results which were acquired using datasets \datasetB and \datasetC suggest that the negative effect of P-fairness optimization on user-fairness often tends to be stronger than the effect in the other direction. In \datasetB, consider the deterioration of DCF $0.3797 \rightarrow 0.3877$ when we are focused on P-Fairness optimization, while when focused on C-fairness we see DPF is not only worsened but also improved $0.8404 \rightarrow 0.8382$. On the other hand, on \datasetC we can see the worsening of DPF i.e., $0.45760 \rightarrow 0.45525 $.

Overall, these result which are obtained on \datasetB and \datasetC suggest that one-sided P-fairness optimization tend to harm consumer fairness more than the other way around. We deem this to be related to the segmentation threshold to determine active vs. non-active users. Given that, active users account for only 5\% of the user population, optimizing for user-relevance have little impact on the exposition of item groups.  When combined consumer-producer fairness bias is concerned, as measured with $mCPF$, we can see the effects more clearly in the boxplots in Figure \ref{fig:CPevalBoxPlot}.  According to the results, the general trend is that optimizing on P-fairness would result in a more considerable reduction of the overall fairness, mCPF in comparison to C-fairness optimization.

\partitle{\textit{Answer to Q3.}}
This question can be answered by observing the pattern in $\Delta(\%)$ values, \ie  the improvement of $mCPF$ in comparison to the fairness-unaware N model. The aggregate mean values of $\Delta(\%)$ across all baseline algorithms for types C, P, CP on \datasetB and \datasetC respectively are (3.83\%, 9.25\%, 12.84\%) and (10.36\%, 31.28\%, 46.90\%) respectively. Moreover, the average value for $mCPF/All$ in (N-model, CP-fairness) model in Epinion and Gowalla are (18.73, 16.53) and (7.31, 4.51). To further compare the overall quality of CP-fairness across the baseline algorithms, we can observe the normalized user-item biases over accuracy defined as $mCPF/All$ for (PF, WMF, NeuMF, VAECF) which are (13.0160, 31.3632, 14.5451, 13.1754) and (0.1343, 2.8952, 9.3087, 5.2617) for Epinion and Gowalla dataset respectively.

\partitle{\underline{\textbf{Observation.}}}
Investigating $\Delta(\%)$ values, it is vivid that our CP-fairness method has the ability to significantly reduce the fairness disparity between user and item groups, jointly evaluated by mCPF metric, when compared with single-sided fairness regardless of the baseline recommendation algorithm. Another observation by comparing $mCPF/All$ values among the CP-fairness model and N baselines is that the proposed CP-fairness model achieves this improvement without sacrificing the overall accuracy of baselines. Among the baseline algorithms, PF is showing the best performance for balancing fairness and accuracy aspects on both Gowalla and Epinion datasets shown from the $mCPF/All$ values. Further, it can be seen that although NeuMF and VAECF are capable of reaching higher overall accuracy, this comes with the cost of significant disparity among producers and amplifying the existing data bias.  

To further evaluate the effectiveness of our proposed CP-fairness re-ranking algorithm in enhancing the user and item fairness requirement in a two-sided marketplace, we performed an extensive experiment (128 experimental cases) on eight real-world datasets. Figure \ref{fig:CPevalBoxPlot} shows the distribution of mCPF performance among all datasets using the four baseline algorithms, \ie PF, WMF, NeuMF, and VAECF. As shown in Figure \ref{fig:CPevalBoxPlot}, in all datasets, the CP-fairness model achieves lower mCPF compared to baselines and unilateral fairness models. It is worth noting that unilateral fairness models show different behavior across datasets which is due to the underlying data characteristic. For example, in Figure \ref{fig:CPevalBoxPlot_MovieLens}, the C-fairness method achieves better results in comparison to P-fairness, while in Figure \ref{fig:CPevalBoxPlot_Fourquare} on the Foursquare dataset, the results on the P-fairness approach is much better than the C-fairness approach.
Furthermore, the experiment results demonstrate that our CP re-ranking algorithm can not only shrink the fairness disparity between the two groups of users and items but also provide overall competitive accuracy (nDCG@10) in comparison to the unilateral fairness methods and baselines as reported on top of the plots in Figure \ref{fig:CPevalBoxPlot}.

\partitle{Ablation Study.}
\label{sec:ablation}
We also analyze the impact of the optimization regularization parameters $\lambda_1$ and $\lambda_2$ from Equation \ref{eq:optimization}. We expect that the bigger $\lambda_1$ and $\lambda_2$ are, the fairer our proposed model will be. However, the excessive pursuit of fairness is unnecessary and could have an adverse impact on the overall recommendation performance. Therefore, we are interested in studying how different values of $\lambda_1$ and $\lambda_2$ in Equation \ref{tbl:results_epinion} can influence the total performance of the system, \ie total accuracy, and novelty of the recommendation, as well as their impact on consumer and provider groups.

The results presented in Figure \ref{fig:ablationGowallaWMF} indicate that $\lambda_1$ exhibits a more \dquotes{accuracy-centric} behavior, implying that this parameter affects user-fairness and overall system accuracy with little impact on the items' expositions. 
This can be an indication that active and inactive users do not have much difference in the type of items they use (in terms of popularity). On the other hand, $\lambda_2$ exhibits an \dquotes{exposure-centric} that can impact both the beyond-accuracy and accuracy of the system. The reason why an excessive rise in $\lambda_2$ does not benefit us of the overall accuracy can be explained by the fact that many long-tail items lack sufficient preference scores (interactions), making it unclear whether including these items in recommendation lists will actually result in user appreciation. In general, striking the right balance when selecting model parameters is critical for increasing the fairness and overall utility of the marketplace.
\section{Conclusions}
\label{sec:conclusion}
In this paper, we study fairness-aware recommendation algorithms from both the user and item perspectives. We first show that current recommendation algorithms produce unfair recommendation between different disadvantaged user and item groups due to the natural biases and imperfections in the underlying user interaction data. To address these issues, we propose a personalized CP-fairness constrained re-ranking method to mitigate the unfairness of both user and item groups while maintaining the recommendation quality. Extensive experiments across eight datasets indicate that our method can reduce the unfair outcomes on beneficiary stakeholders, consumers and providers, while improving the overall recommendation quality. This demonstrates that algorithmic fairness may play a critical role in mitigating data biases that, if left uncontrolled, can result in societal discrimination and polarization of opinions and retail businesses.


\bibliographystyle{ACM-Reference-Format}
\bibliography{references}

\end{document}